\documentclass[11pt,a4paper]{article}

\usepackage[disable]{todonotes}

\usepackage{style}
\usepackage{shortcuts}
\usepackage{lowco2}

\title{Distributed Sparsest Cut via Eigenvalue Estimation\footnote{This research was funded in whole or in part by the Austrian Science Fund (FWF) \url{https://doi.org/10.55776/P36280} and \url{https://doi.org/10.55776/I6915}. For open access purposes, the author has applied a CC BY public copyright license to any author-accepted manuscript version arising from this submission.}}
\author{Yannic Maus \textcircled{r}\footnote{The author ordering was randomized using \url{https://www.aeaweb.org/journals/policies/random-author-order/} generator. It is requested that citations of this work list the authors separated by \texttt{\textbackslash textcircled\{r\}} instead of commas.}\hspace{.5em} Tijn de Vos}
\date{}

\begin{document}

  \begin{titlepage}
    \maketitle
    \thispagestyle{empty}
   
    \begin{abstract}
        We give new, improved bounds for approximating the sparsest cut value or in other words the conductance $\phi$ of a graph in the \CONGEST model. As our main result, we present an algorithm running in $O(\log^2 n/\phi)$ rounds in which every vertex outputs a value $\tilde \phi$ satisfying $\phi \le \tilde \phi \le \sqrt{2.01\phi}$. In most regimes, our algorithm improves significantly over the previously fastest algorithm for the problem [Chen, Meierhans, Probst Gutenberg, Saranurak; SODA~25]. Additionally, our result generalizes to $k$-way conductance. 

We obtain these results, by approximating the eigenvalues of the normalized Laplacian matrix $L:=I-\Deg^{-1/2}A\Deg^ {-1/2}$, where, $A$ is the adjacency matrix and $\Deg$ is the diagonal matrix with the weighted degrees on the diagonal. 
We show our algorithms are near-optimal by proving a lower bound for computing the smallest non-trivial eigenvalue of $L$, even in the stronger LOCAL model

The previous state of the art sparsest cut algorithm is in the technical realm of expander decompositions.
Our algorithms, on the other hand, are relatively simple and easy to implement. At the core, they rely on the well-known \emph{power method}, which comes down to repeatedly multiplying the Laplacian with a vector. This operation can be performed in a single round in the \CONGEST model.

All our algorithms apply to weighted, undirected graphs. Our lower bounds apply even in unweighted graphs. 
    \end{abstract}
\vfill
    \lowcotwo\ This is a low-co2 research paper: \lowcotwourl[\lowcotwoversion]. This research was developed, written, submitted and presented without the use of air travel.

    \paragraph*{Acknowledgements}
    We would like to thank Leo Wennmann, for the many hours of discussion on this topic. We would like Joachim Orthaber and Malte Baumecker, for thinking along. We would like to thank Faith Ellen, Sebastian Brandt, Alexandre Nolin, and Eva Rotenberg, for the initial discussions that were the inspiration for this research. We would like to thank Schloss Dagstuhl, since the research was initiated during the Dagstuhl seminar \textit{Graph Algorithms: Distributed Meets Dynamic}.
    \newpage
     \tableofcontents
    \thispagestyle{empty}    
     \newpage
      \listoftodos
  \end{titlepage}
  
  \newpage

\clearpage
\setcounter{page}{1}
\section{Introduction}
A graph cut is a partition of the vertices into two disjoint subsets, where the size of the cut is the total number of edges crossing between them, and the sparsity of the cut measures how small this cut is relative to total edge weight of the smaller side, i.e., in comparison to the sum of its vertex degrees.  
Graph cuts are a central object in theoretical computer science and algorithm design and have been extensively studied in various settings and models of computation, both explicitly studied, e.g., in \cite{Karger00,AroraRV09,SGP09,SarmaMPU15,KuhnM15,DoryEMN21,CKGS23}, and as central subroutines, e.g., in \cite{SpielmanT04,SaranurakW19,ChenMGS25}.

Cuts are especially important in communication networks because they represent the capacity for information to flow between different parts of the system. A sparse cut—where only a few edges connect two large subsets—can become a bottleneck, severely limiting bandwidth or throughput. A central question for guaranteeing the performance of network algorithms is to determine the minimum sparsity of cuts in a given communication network. For example the runtime of the famous PUSH-PULL protocol depends on the sparsest cut value~\cite{G11}. 
More generally, graphs without sparse cuts support fast mixing of random walks~\cite{SarmaNPT13}, and enable efficient routing due to their strong connectivity~\cite{GhaffariKS17,GhaffariL18}. 
This motivates the use of expander decompositions, which partition a graph into \emph{high-conductance} clusters (i.e, no sparse internal cuts), and sparse inter-cluster connections.  
Such decompositions are powerful algorithmic tools: thanks to efficient routing, many algorithms run significantly faster on high-conductance subgraphs, allowing problems to be solved locally within clusters before handling the sparse inter-cluster parts. This strategy has been used in many applications, see, e.g., \cite{ChangPZ19,ChangS19,DagaHNS19,EdenFFKO22,Censor-HillelGL20}.
The distributed expander decompositions~\cite{ChangPZ19,ChangS19,ChangS20,ChenMGS25} are all constructed by repeatedly computing sparsest cuts. 


\medskip

In this work, we present simple, time and bandwidth efficient distributed algorithms to approximate the sparsest cut value. Our randomized algorithms run in the classic \CONGEST model of distributed computing where a communication network is abstracted as an $n$-vertex graph with vertices being computing entities and edges serving as bandwidth-limited communication links. Communication happens in synchronous rounds, in each of which each vertex can send a $O(\log n)$-bit message to each of its neighbors   and the complexity measure is the number of communication rounds until the vertices have computed their outputs, e.g., until they have output an approximate value of the sparsest cut.

Formally, the  \emph{sparsity} of a set (or cut) $S\subseteq V$ in a graph $G=(V,E)$ is defined as as 
\begin{equation*}
    \phi(S) := \frac{w(E(S,V\setminus S))}{\min\{\Vol(S),\Vol(V\setminus S)\}}~,  
\end{equation*}
where $w(T):= \sum_{e\in T}w(e)$ for any edge set $T$ and the volume $\Vol(S)$ of a set of vertices $S$ is the sum of the weighted degrees of the vertices in the set: $\sum_{v\in S}\sum_{uv\in E}w(uv)$. The value $\phi(S)$ is also known as the \emph{conductance} of the set $S$.  The \emph{sparsest cut $\phi(G)\in [0,1]$} of a graph $G$ is the minimum sparsity of a subset $S\subseteq V$. This is also known as the \emph{conductance} of the graph. 

The relative notion $\phi(G)$ 
can be viewed as a normalized quantification of cut size. 
In contrast to computing minimum cuts, which quantify the absolute size of a cut, determining the sparsest cut or the conductance is not only NP-complete~\cite{MatulaS90} but also APX-hard under the unique games conjecture~\cite{ChawlaKKRS06}. Consequently, although no polynomial-time algorithm can approximate the optimum arbitrarily well (unless $P = NP$), a variety of approximation algorithms with provable---albeit weaker---guarantees have been developed across different computational models.
Next, we present our results and compare them with the state of the art in the field. As the sparsest cut and the conductance have been extensively studied we provide further background and related work in \Cref{ssec:relatedWork}.

\subsection{Our Contribution on Sparsest Cut Approximation}
Our first result is an efficient \CONGEST algorithm to approximate the sparsest cut. 

\begin{restatable}{theorem}{CONGESTSparsestCut}
\label{thm:CONGESTSparsestCut}
     There is a randomized \CONGEST algorithm that, given an undirected weighted $n$-node graph with sparsest cut $\phi$, w.h.p.\ gives an approximation of the sparsest cut in $O(\log^2 n/\phi)$ rounds. In particular, every vertex outputs $\tilde \phi$ such that $\phi \le \tilde \phi \le \sqrt{2.01\phi }$. 
\end{restatable}
In general, there are two approaches for developing  approximation algorithms for the sparsest cut, 
the \emph{flow approach}, that we detail in \Cref{ssec:relatedWork}, and the \emph{spectral approach}, which we  use to prove \Cref{thm:CONGESTSparsestCut}. On a high-level, one computes the spectrum (or parts of it) of the graph, i.e., its eigenvalues, and then uses Cheeger's Inequality\footnote{Cheeger's Inequality~\cite{AlonM85,Alon86,SinclairJ89} relates the second-smallest eigenvalue $\lambda_2$ of the normalized Laplacian to the graph's conductance $\phi$, showing that $\tfrac{1}{2} \lambda_2 \le \phi \le \sqrt{2\lambda_2}$. See \Cref{thm:Cheeger} for the formal statement. }
to derive an approximation of $\phi(G)$ from the eigenvalues. Due to the limitations imposed by Cheeger's Inequality, the best possible approximation achievable through spectral methods is $\sqrt{2\phi}$. Our technical result is slightly stronger than \Cref{thm:CONGESTSparsestCut}. In fact, we can attain this theoretical bound -- up to an arbitrarily small additive error -- in $O(\log^2 n /\eps^2+\log n/\eps^4)$ rounds. See the proof of \Cref{thm:CONGESTSparsestCut} for the details. 
 


For unweighted graphs, we compare our result against the previously best algorithm by Chen, Meierhans, Probst Gutenberg, and Saranurak that provides a $O(\phi \log^2n)$-approximation of $\phi$ \cite{ChenMGS25}\footnote{They do not state this result explicitly, but they implement the cut-matching game from \cite{SaranurakW19} (see Appendix A) in the distributed setting. As stated in \cite{SaranurakW19}, this cut-matching game gives such an approximation.} Our approximation is more precise than \cite{ChenMGS25} when $\phi =\Omega(1/\log^4 n)$, and is faster in all cases: \cite{ChenMGS25} takes $O(\poly \log n/\phi^4)$ rounds. Upon inspection of their $\poly \log n$, they have a total round complexity of at least $O(\log^{20} n)$ in the regime where their approximation is more precise. 

Prior to \cite{ChenMGS25}, Das Sarma, Molla, and Pandurangan~\cite{SarmaMP15} gave a different algorithm to approximate the sparsest cut, which also works for weighted graphs. They computed an  $O(\sqrt{\phi}\poly\log n)$-approximation in $O(\tfrac{1}{b}(n+\tfrac{1}{\phi})\log^2 n)$ rounds if the graph has a sparse cut of \emph{balance}~$b$, i.e., both sides $S$ of the cut satisfy $\Vol(S)\ge 2b|E|$. 
Building on the methods of \cite{SarmaMP15}, Kuhn and Molla improved the approximation ratio to $O(\sqrt{\phi \log n})$ and the runtime to 
 $O(D+\tfrac{\log^2 n}{b\phi})$ \cite{KuhnM15}. For balanced cuts $b=\Omega(1)$, we improve their approximation ratio by a factor $\sqrt{\log n}$, while matching their running time. For unbalanced cuts, also their algorithm can take as much as $\Omega(n^2)$ time, 
 so we do not only improve the approximation ratio, but also the running time exponentially. 

 We note that \cite{ChenMGS25,SarmaMP15} and \cite{KuhnM15} do not only provide an approximate cut value, but also give a corresponding cut.

\paragraph{$k$-Way Sparsest Cut.}
There exists a natural generalization of the sparsest cut to \emph{multi-cuts}~\cite{LeeGT12}: the \emph{$k$-way sparsest cut} or \emph{$k$-way conductance} $\phi_k$ is defined follows:
    \begin{equation*}
        \phi_k := \min_{V_1, V_2, \dots, V_k\subset V} \max_{i=1, \dots, k} \frac{|E(V_i,V\setminus V_i|}{\Vol(V_i)},
    \end{equation*}
    where the minimum is taken over disjoint, non-empty $V_i$.
Note that $\phi_2=\phi$. The $k$-way conductance captures how well-connected a graph is across multiple parts; it can describe multi-partition bottlenecks, which standard conductance may miss. This is in particular relevant for real-world networks that usually consists of multiple clusters.

We prove the following theorem for approximating the $k$-way sparsest cut. 
\begin{restatable}{theorem}{CONGESTkWaySparsestCut}
\label{thm:CONGESTkWaySparsestCut}
     There is a randomized \CONGEST algorithm that, given a constant $k\geq 2$ and undirected weighted graph with $k$-way sparsest cut $\phi_k$, w.h.p.\ gives an approximation of the $k$-way sparsest cut in $O(\log^2 n\poly(\phi_k^{-1}))$ rounds. In particular, every vertex outputs $\tilde \phi_k$ such that $\phi_k \le \tilde \phi_k \le O(\sqrt{\phi_k})$. 
\end{restatable}
To the best of our knowledge, we are the first to study the $k$-way conductance in the distributed setting. 

\paragraph{Lower bounds on approximating sparse cuts.}
Generally, it is known that approximating the sparsest cut takes $\Omega(D)$ rounds. In particular, Das Sarma, Molla, and Pandurangan~\cite{SarmaMP15} show a $\tilde \Omega(D+\sqrt n)$ \CONGEST lower bound for any multiplicative approximation in weighted graphs. They  reduce the spanning connected subgraph problem to a sparsest cut instance with edge weights that are either $0$ or $1$. This means that the lower bound shows that it is hard to decide whether $\phi=0$ holds or not. For $\phi$ bounded away from zero, there is no such lower bound.  Our upper bound of $O(\log^2n/\phi)$ rounds, \Cref{thm:CONGESTSparsestCut}, indeed takes $D$ rounds, since $D=O(\log n/\phi)$. However, for $\phi =\omega(1/\sqrt{n})$, the upper bound does not have a $\sqrt{n}$-term. 

We prove a $\Omega(D)$ lower bound that holds even for unweighted graphs and in the \LOCAL model. \LOCAL is identical to \CONGEST, except that there is no bound on the message size. 
Even for a class of fixed diameter, any algorithm needs $\Omega(D)$ time for a $O(n)$-approximation.

 \begin{restatable}{theorem}{SparsestCutMultLB}\label{thm:Sparsest_cut_mult_LB}
    Let $D\le n/2$ be a parameter and consider the class of graphs with diameter $D$. 
    An algorithm that with probability $>\tfrac{1}{2}$ satisfies that every vertex~$v$ of a graph $G$ outputs a value~$\tilde \phi_v\ge \phi(G)$, for sparsest cut $\phi(G)$, and at least one vertex has $\phi_v\le \Theta(n\cdot \phi)$ needs $\Omega(D)$ rounds in \LOCAL. 
\end{restatable}

Note that the $\Omega(D)$ lower bound in \Cref{thm:Sparsest_cut_mult_LB} holds even in the case where only a single vertex needs to witness the cut.

The above proof is about discerning graphs with sparsity $1/n$ and $1/n^2$. If \emph{every vertex} needs to output a correct estimation, a diameter lower bound exists for larger values of $\phi$. 
 \begin{restatable}{theorem}{SparsestCutAddLB}
 \label{thm:SparsestCutAddLB}
    Let $0<\eps<1$ be a constant and let  $\phi>1/n^{1-\eps}$.
    An algorithm that with probability $>\tfrac{1}{2}$ satisfies that every vertex~$v$ outputs a value~$\tilde \phi_v$ such that $\phi \le \tilde \phi \le O(n^{1+\eps}\cdot \phi)$, where $\phi$ is the sparsest cut value of the graph needs at least $\Omega(D)$ rounds in the \LOCAL model. 
 \end{restatable}
  Our \CONGEST upper bound for approximating the sparsest cut w.h.p.\ is $O(\log^2 n/\phi)$ -- only a factor $\log n $ away from optimal, since the graphs from the above observation have sparsest cut $\phi$ and diameter $\Theta(\log n/\phi)$~\cite{ChierichettiGLP18}. We note that with constant probability, we match the lower bound, see \Cref{ssec:techOverview}. 

\subsection{Our Contributions on Spectrum Approximation}
Our core technical contributions are efficient \CONGEST algorithms to compute the approximate (partial) spectrum of a graph. In essence, \Cref{thm:CONGESTSparsestCut,thm:CONGESTkWaySparsestCut} then follow as direct consequences via Cheeger's Inequality. Nevertheless, we chose to present these theorems first, as they are more accessible to a broader audience and can be compared with prior work.

\paragraph{Eigenvalues of the Normalized Laplacian.}
The normalized Laplacian of a graph is the matrix defined as $L:=I-\Deg^{-1/2}A\Deg^{-1/2}$, where $A$ is the adjacency matrix and $\Deg$ is the (weighted) degree matrix.  See \Cref{sec:Laplacian_preliminaries} for more details. Its spectrum captures many central properties of the graph, e.g., various types of its connectedness like $k$-way conductance~\cite{LeeGT12} and small set expansion~\cite{AroraBS15}, or how close it is to being bipartite~\cite{BauerJ09,liu2015multi}.  We provide the following theorem to compute an approximation of that spectrum.

\begin{theorem}[Eigenvalue estimation]
\label{thm:EigenvaluesMain}
Let $\varepsilon > 0$ be a parameter and $k\in \mathbb{N}$ be a constant. There are randomized \CONGEST algorithms that, given an undirected, weighted graph ${G=(V,E)}$, with high probability approximate the eigenvalues $\lambda_1\leq \lambda_2\leq \ldots \leq \lambda_n$ of the graph's normalized Laplacian as follows:

\begin{enumerate}
 \item $\lambda_1, \ldots,  \lambda_k$:  an additive $\varepsilon$-approximations in $O(\log^2 n \cdot \mathrm{poly}(\varepsilon^{-1}))$ rounds, 
    \item  $\lambda_2$: an additive $\varepsilon$-approximation in $O\left(\frac{\log^2 n}{\eps} + \frac{\log n}{\eps}\cdot \left\lceil \frac{\lambda_2}{ \eps}\right\rceil\right) = O\left(\frac{\log^2 n}{\eps} + \frac{\log n}{\eps^2}\right)$ rounds, 
  \item $\lambda_n$: a multiplicative  $(1 \pm \varepsilon)$-approximation  in $O\left(D + \tfrac{\log^2 n}{\varepsilon}\right)$ rounds.
   \end{enumerate}
\end{theorem}
Eigenvalue $\lambda_2$ is the eigenvalue approximating the sparsest cut via Cheeger's Inequality. It could  either be approximated via the second algorithm, or slightly slower via the first.

For the smaller eigenvalues, \Cref{thm:EigenvaluesMain} states additive approximations avoiding the diameter dependence, e.g., when $\eps$ is an arbitrarily small constant \Cref{thm:EigenvaluesMain} approximates $\lambda_2$ in $O(\log^2 n)$ rounds even though the diameter may be up to linear in $n$. 

As $\eps$ can be arbitrary in \Cref{thm:EigenvaluesMain}, we can obtain multiplicative approximations by setting  $\eps = \eps'\lambda_2$ for any $\eps'>0$. However, this does not circumvent the diameter lower bound of \Cref{thm:eigenvalue_LB} presented below. By \Cref{thm:Cheeger} and \Cref{lm:diam_k_sparsest_cut}, $D=O(\log n/\lambda_2)$ holds for any graph, so the diameter term is then always asymptotically smaller than the overall runtime of the algorithm.

The main prior work for distributed eigenvalue approximation is by Kempe and McSherry~\cite{KempeM08}. They provide a distributed algorithm that, for constant $k$, approximates the $k$-largest/smallest eigenvalues in $O(\tau_{\rm{mix}}\log^2n)$ rounds. 
We note that the mixing time can be as large as $n^{\Omega(1)}$. This happens when $\lambda_2$ is really small (in which case our algorithms are also slow), but also when the graph is (close to) bipartite. In the latter case, our algorithms can be exponentially faster. 
Furthermore, an important caveat here is that Kempe and McSherry do not obtain their result in the \CONGEST model: they do not discuss the bandwidth constraint. They do specify that they only need $\poly(k)$ messaged per edge per round. However, they do not take the \emph{size} of these messages into account. In our work, this is a non-trivial contribution, as intermediate values may grow fast. We think similar techniques might be able to work for their approach as well, but this requires non-trivial adjustments and additional analysis.

Later, Becchetti et al.~\cite{BecchettiCNPT20} gave an algorithm for finding \emph{balanced} cuts, where our algorithm works for \emph{all} cuts. We provide more details on this algorithm in \Cref{ssec:techOverview}. 

An alternative approach to distributed eigenvalue computation is via computing a spectral sparsifier, i.e., a sparse subgraph $H$ with $\tilde O(n/\eps^2)$ edges that approximately preserve the graph's spectrum. Next, the sparsifier $H$ can be gathered at one vertex of the graph that can locally compute the eigenvalues of $H$ and hence obtain approximations of all eigenvalues of $G$.  Computing $H$ is actually relatively fast and only requires $\poly(\log n,\eps^{-1})$ rounds~\cite{KX16}. Collecting $H$ at a single vertex is the expensive part, requiring $\tilde O(n/\eps^2)$ rounds. This is exponentially slower than \Cref{thm:EigenvaluesMain}, though this more general method immediately computes the full spectrum while we focus on single eigenvalues.  


\medskip

We note that all our upper bounds even hold in the more restricted Broadcast \CONGEST model in which a vertex needs to send the same message to all of its neighbors in each round. All our algorithms become a $\log n$-factor faster if we only ask for constant probability as typical in property testing, see the discussion in \Cref{ssec:techOverview}.

\paragraph{Lower Bound.}
All of our multiplicative upper bounds take diameter time. This means that at the end of the algorithm, every vertex can easily output the same eigenvalue or sparsest cut estimate. One may ask whether one can obtain faster algorithms if the  output is more local. Perhaps every vertex outputs the sparsest cut it sees in some neighborhood, and there is a guarantee that at least one vertex will see an approximately sparsest cut. In \Cref{thm:Sparsest_cut_mult_LB}, we showed that this is not true for sparsest cut approximation. We also show that approximating $\lambda_2$ at even one node requires $\Omega(D)$ rounds. 

\begin{restatable}{theorem}{EigenvalueLB}\label{thm:eigenvalue_LB}
    Any algorithm that with probability $>\tfrac{1}{2}$ ensures that all vertices output a value $\tilde \lambda_2\ge \lambda_2$, and at least one vertex additionally satisfies $\tilde \lambda_2 \leq 2\lambda_2$ requires $\Omega( D)$ rounds. 
\end{restatable}
An $\Omega(D)$ lower bound could also follow from our sparsest cut lower bounds via Cheeger's Inequality, but such an approach incurs an inherent quadratic loss due to the approximation gap in Cheeger's.  

\paragraph{Decision problem and connections to property testing.} For the lower bounds in \Cref{thm:Sparsest_cut_mult_LB} (and similarly for \Cref{thm:eigenvalue_LB}), the vertices that do not output a correct value, still need to output a value $\tilde \phi\ge \phi$. If we do not do this, then a trivial algorithm where every vertex outputs a random number between $0$ and $1$ will be correct most of the time. In other words, we aim to solve the decision problem: ``is there a cut sparser than $\phi$?'' If all vertices say ``no'', the answer is no. If at least one vertex answers ``yes'', the answer is yes. As we are only solving the approximate version of the question this is closely related to the concept of \emph{property testing} in \CONGEST, introduced in \cite{Censor-HillelFS19} and extensively studied since then, e.g, \cite{EFFGLMMOORT17,LMD18}. ``Conductance property testing'' has been explicitly studied by \cite{FY18} and \cite{BTT24} and \Cref{thm:CONGESTSparsestCut} improves significantly on their results. We detail the relation to property testing and their results in \Cref{ssec:relatedWork}.

Also the eigenvalues of the adjacency matrix and the non-normalized Laplacian describe many graph properties. Among other properties, they relate to the size of independent sets (Hoffman's bound)~\cite{haemers1995interlacing}, the chromatic number~\cite{wilf1967eigenvalues}, the average density of cuts~\cite{haemers1995interlacing}, the toughness of the graph~\cite{liu2010algebraic}, Hamiltonicity~\cite{fiedler2010spectral}, the matching number~\cite{GuL22}, and the existence of a perfect matching~\cite{brouwer2005eigenvalues}.

\paragraph{Approximating the Entire Spectrum.}  
Via similar methods, one can compute a multiplicative approximation of the entire spectrum. However, the output will not reveal how many eigenvalues there are of a certain approximation, i.e., it gives the spectrum \emph{without multiplicity}. The running time of such an algorithm is $n^{o(1)}(\sqrt n+D)/\eps^2$ rounds, see \Cref{sec:approx_spectrum} for details. Note that in particular this gives a multiplicative approximation of the smallest non-zero eigenvalue~$\lambda_2$. In case $\lambda_2$ is small -- roughly $\lambda_2= o( 1/\sqrt n)$ -- this is faster than \Cref{thm:EigenvaluesMain}.

\subsection{Related Work} \label{ssec:relatedWork}
\paragraph{Sparsest Cut.}
As mentioned there are two approaches to solving the sparsest cut problem. The \emph{spectral approach} used in this paper and the \emph{flow approach}. The spectral approach encompasses techniques like Cheeger's Inequality, spectral sparsifiers (discussed above), and also random walks. For instance, the results of \cite{das2015extremal,KuhnM15,FY18,BTT24} rely on random walks, which are considered spectral in nature since their convergence analysis depends on spectral properties.  This approach is inherently limited by Cheeger's Inequality, capping the approximation factor at $\sqrt{2\phi}$, the guarantee we match in this work. 

The flow approach, based on the \emph{cut-matching game}, computes a number of flows that will either guarantee expansion or find a sparse cut. Theoretically, this can give a $(\phi \poly \log n)$-approximation. In the sequential setting, the state-of-the-art is a $O(\phi\sqrt{ \log n})$-approximation by Arora, Rao, and Vazirani~\cite{AroraRV09}. In the distributed setting, the state-of-the-art is the aforementioned $O(\phi\log^2 n)$-approximation~\cite{ChenMGS25}. Before this recent result, the flow based gave worse approximations: \cite{ChangPZ19,ChangS19} gave a $\tilde \phi \le O(\phi^{1/3}\log^{5/3}n)$-approximation in $O(D\cdot \poly \log n,\phi^{-1})$ rounds, for $\phi = O(1/\log^5 n)$. Additionally, Chang and Saranurak~\cite{ChangS20} gave a deterministic version of that result with additional sub-polynomial factor in the round complexity.

\paragraph{Property Testing.}
In our paper, we are essentially solving the decision problem: ``does $G$ have sparsest cut at least $\phi$?''
In property testing, the goal is to solve a relaxed variant of the corresponding decision problem: if the input network satisfies a property, then, with constant probability, all the vertices accept
but if the input network is $\eps$-far\footnote{This refers to $\eps m$ changes to the graph. For a precise definition, see, e.g., \cite{BTT24}.} from satisfying the property, then at least one
vertex rejects. Testing for sparsest cuts has explicitly been targeted in this setting by \cite{FY18} and \cite{BTT24}. For a fixed $\phi$, 
\cite{BTT24}, improving on \cite{FY18}, solves this question as follows. All vertices output ``yes'', with probability at least $2/3$ if the sparsest cut is at least $\phi$, and at least one vertex outputs ``no'', with probability at least $2/3$ if $G$ is $\eps$-far from having sparsest cut above $\phi^2/2880$. Similar to our result, the $\phi^2$ term stems from Cheeger's Inequality. Our results improve upon this in several ways. For achieving a constant probability, we only need $O(\log n/\phi)$ rounds, as compared to $O(\tfrac{\log n}{\eps\cdot \phi^2})$ rounds in their algorithm. Most notably, we do not solve the property testing version, and do not require that the instance is $\eps$-far from having a sparse cut above $O(\phi^2)$. Thus, in some sense we are solving their version with $\eps=1/n^2$. Additionally, our results extend to weighted graphs and the constant in our approximation factor is significantly better.  Lastly, we also approximately compute the sparsest cut value instead of only solving the decision problem. 


\paragraph{Cut Sparsifiers.}
We have already discussed how to use spectral sparsifiers to approximate the spectrum in $\tilde O(n/\eps^2)$ rounds~\cite{KX16}. They can also directly be used for approximating sparsest cuts without going through Cheeger's Inequality: $(1\pm\eps)$-approximate spectral sparsifiers are $(1\pm\eps)$-cut sparsifiers. This means that in $\tilde O(n/\eps^2)$ rounds, we can collect a subgraph $H$ in one vertex, in which \emph{all} cuts are at most a factor $(1\pm\eps)$ away from the cut in $G$. In particular, we can find a $(1\pm \eps)$-approximation of the sparsest cut, and more generally a $(1\pm \eps)$-approximation of the $k$-way sparsest cut in $\tilde O(n/\eps^2)$ rounds.


\paragraph{Laplacian Paradigm.}
The \emph{Laplacian paradigm} is an umbrella term for algorithmic techniques relying on the Laplacian and its properties.  This line of research was initiated by Spielman and Teng~\cite{SpielmanT04}, who showed that linear equations in the Laplacian matrix of a graph can be solved in near-linear time. Algorithms within the Laplacian paradigm combine numerical techniques -- from the matrix point of view -- with combinatorial techniques -- from the graph point of view. Nowadays, it encompasses techniques like solving Laplacian systems, spectral sparsifiers, electrical flow, effective resistance, expander decompositions, continuous optimization, interior-point methods, gradient descent, and preconditioning. This has had many applications, including but not limited to flow problems~\cite{Madry13,Sherman13,KelnerLOS14,Madry16,Peng16,CohenMSV17,LiuS2020FasterDivergence,LiuS20,AxiotisMV20}, bipartite matching~\cite{BrandLNPSS0W20}, and (parallel) shortest paths~\cite{Li20,AndoniSZ20}. 
 
Also in the distributed world, this has booked many successes, e.g., ~\cite{GhaffariKKLP18,BeckerFKL21,ForsterGLPSY21,Anagnostides0HZ22,ForsterV22,ForsterV23,Vos23}. With this paper, we hope to contribute to a better understanding of the Laplacian matrix and its applications in a distributed setting.

\paragraph{Comparison to the Min-Cut Problem.}
In the minimum cut (min-cut) problem, the goal is to find a cut $S$ of minimum value $|E(S,V\setminus S)|$. Although the problem statement seems similar to the sparsest cut problem, it has a very different computational complexity. As said, in the sequential setting, it is NP-hard to approximate the sparsest cut (under the unique games conjecture)~\cite{MatulaS90,ChawlaKKRS06}. On the other hand, there are linear time algorithms for computing the exact min-cut~\cite{Karger00}. 

In the distributed setting, we show that the $\Omega(\sqrt{n}+D)$ lower bound does not apply for approximating the sparsest cut problem. The min-cut problem however has a complexity of $\tilde \Theta(\sqrt n+D)$~\cite{SarmaHKKNPPW12,DoryEMN21}, 
where the lower bound holds even for a $\poly(n)$-approximation. 

\subsection{Technical Overview}
\label{ssec:techOverview}

As mentioned, \Cref{thm:CONGESTSparsestCut,thm:CONGESTkWaySparsestCut} follow as direct consequences from our eigenvalue estimation results via Cheeger's Inequality. Here, we focus on our core technical contribution: efficient \CONGEST algorithms to compute the approximate (partial) spectrum of a graph. 
The main ingredient of our algorithms is the well-known \emph{power method} for approximating the largest eigenvalue of a matrix. 
\paragraph{The Power Method Algorithm.}
Let $M$ be any matrix. The power method, also known as \emph{power iteration} or \emph{Von Mises iteration}, starts with a random vector and converges to the eigenvector corresponding to the largest eigenvalue of $M$. It is a standard technique in numerical linear algebra, see, e.g., \cite{golub2013matrix}. The procedure is as follows. 
\begin{tcolorbox}[title=\textbf{Power Method}, colback=gray!10!white, colframe=gray!70!white, coltitle=black, fonttitle=\bfseries]
\begin{algorithmic}[1]
    \State \textbf{Input:} Matrix $M \in \mathbb{R}^{n \times n}$, number of iterations $k$.
    \State \textbf{Initialize:} A random vector $x_0 \in \mathbb{R}^n$.
    \For{$i = 1$ to $k$}
        \State $x_i \gets M x_{i-1}$.
    \EndFor
    \State \textbf{Output:}
    \begin{itemize}
        \item Approximate dominant eigenvalue: $\lambda \approx \frac{x_k^T M x_k}{x_k^T x_k}$;
        \item Approximate dominant eigenvector: $x_k$.
    \end{itemize}
\end{algorithmic}
\end{tcolorbox}

With constant probability over the randomness of the start vector $x_0$, the power method outputs a vector $x_k$ whose \emph{Rayleigh coefficient} $\frac{x_k^T M x_k}{x_k^T x_k}$ is close to the \emph{largest} eigenvalue of the matrix $M$. The main intuition why this happens is as follows: Any vector can be written as a combination of the matrix' eigenvectors and when you apply the matrix multiple times, the contribution from the eigenvector with the largest eigenvalue (in magnitude) dominates. So, over time, the method ``filters out'' all directions except the one corresponding to the largest eigenvalue. This works only if the random start vector $x_0$ has a nonzero component in the dominant eigenvector’s direction; beyond the choice of this initial vector, the process is fully deterministic.


In this paper, we only need the power method for \emph{positive semi-definite (PSD)} matrices, which are symmetric matrices with non-negative eigenvalues. Then all eigenvalues are nonnegative and we can assume the eigenvectors to be orthonormal. The power method also converges slightly faster on PSD matrices -- by a factor $2$ to be precise. 

\begin{restatable}[Power Method]{lemma}{PowerMethod}\label{lm:powerMethod}
Let $M$ be a positive semi-definite matrix with eigenvalues $0\le \mu_1 \le \dots \le \mu_n$.
    Let $x_0\sim \{-1,1\}^n$ be a uniformly random vector and $x_k = M^k x_0$.
    For $\eps \ge 16/n$ and $k\ge \Theta(\log n/\eps)$, it holds with constant probability that 
    \begin{equation*}
        (1-\eps)\mu_n\le\frac{x_k^TMx_k }{x_k^T x_k} \le \mu_n.
    \end{equation*}
\end{restatable}

The power method is suitable for distributed applications, since matrix vector multiplication is easy in the \CONGEST model: this can be done in $O(1)$ rounds where each vertex of the communication network holds an entry of the vector, and each vertex knows the entries of the matrix in its row and column. Important is that the only non-zero entries in a row or column correspond to a neighboring vertex. This is exactly the case when the matrix is closely related to the adjacency matrix -- such as the Laplacian.

Only the final step, of computing the Rayleigh coefficient of $x_k$ takes $O(D)$ time. 
We note that the power method is a widely used routine. Also decentralized implementations have been known for a long time, e.g, \cite{bertsekas2015parallel,penna2014decentralized} or for the special case of PageRank~\cite{SarmaMPU15}. However, to the best of our knowledge, the procedure and its extensions have not been studied in the context of the \CONGEST model. 


One difficulty with implementing the power method in the \CONGEST model as is, is that values grow exponentially with $k$ and hence after a few iterations and they cannot be send to neighbors using $O(\log n)$ bits per message. Fortunately, the method is extremely robust. It still converges quickly after truncating each value after the $O(\log n)$ most significant bits; if one additionally remembers the order of magnitude of the values one can approximately compute the Rayleigh coefficient at the end of the algorithm -- see \Cref{lm:num_stab} for a formal statement.
More involved methods, like the Lanczos method (see, e.g., \cite{golub2013matrix}) might be able to return similar results. However, the Lanczos method is inherently unstable. It is unclear whether this can be adapted to work (efficiently) in the \CONGEST model. 

\paragraph{Averaging Dynamics.}
In community detection, average dynamics is a well-known concept. The idea is that each node updates its value to the average of its neighbors. If $G$ is connected and not bipartite, this always converges. 
The update step is comparable is the same in spirit as a step in the power method, but not entirely equal.
In particular, Becchetti et al.~\cite{BecchettiCNPT20} use average dynamics in the distributed setting.  In case of a sufficiently small balanced cut, and some additional technical assumptions, they can provide an approximate sparsest cut. Since the approximation factor depends on the technical assumptions, we will not state their exact guarantees here. 
With a similar approach, Beccheretti et al.~\cite{BecchettiCNPRT17} can find a sparse cut efficiently if both of the parts do not contain sparse cuts. Moreover, the average dynamics has been studied as a tool for community dynamics in related models that do not directly compare to the \CONGEST model, e.g., the asynchronous, opportunistic communication model~\cite{BecchettiCMNPRT18} and the asynchronous gossip model~\cite{Mallmann-TrennM18}.



\paragraph{The Graph Laplacian Matrix.} We denote $\mathds 1$ for the vector with 1 in each entry and $\sqrt{\deg}=\Deg^{1/2}\mathds{1}$ for the vector with $(\sqrt{\deg})_v=\sqrt{\deg(v)}$ on index (equals vertex) $v$. 

Recall that the normalized Laplacian $L:=I-\Deg^{-1/2}A\Deg^{-1/2}$ has smallest eigenvalue $\lambda_1=0$. 
 This can be seen since $\sqrt{\deg}$ is an eigenvector of $L$ eigenvalue $0$:
\begin{align*}
    L\sqrt{\deg }&=(I-\Deg^{-1/2}A\Deg^{-1/2})\sqrt{\deg } = \sqrt{\deg } -\Deg^{-1/2}A \mathds{1}, 
\end{align*}
Note that $A\mathds{1}$ counts the number of 1s in each row\footnote{In weighted graphs, this is the sum of the weights in each row. Since the degrees are also weighted degrees, the proof remains the same.}, so this exactly sums up to the degree:
   $ L\sqrt{\deg}= \sqrt{\deg } -\Deg^{-1/2}\deg  =0$.
Hence, the smallest interesting eigenvalue to study is $\lambda_2$.

\paragraph{Approximating $\lambda_2$.}
To approximate the smallest non-zero eigenvalue of the Laplacian $L$, we consider the matrix $M:=2I-L$, which has eigenvalues $0\le \mu_1 \le \dots \le \mu_n$, satisfying $\mu_i=2-\lambda_{n-i+1}$. In particular, this means that $\mu_n=2$ and $\mu_{n-1}=\lambda_2$. 
In order to compute  $\lambda_2$, one would like to subtract the components in the direction $v_n=\sqrt{\deg}$, the eigenvector corresponding to the largest eigenvalue $\mu_n$. As a result, the prior largest eigenvalue is shifted to $0$, and the power method converges to the second largest eigenvalue instead.  

As per the discussion above, we know the largest eigenvector $v_n=\sqrt{\deg}$ exactly. We can also subtract the components in this direction exactly. However, a significant facet here is the bit complexity. As we discussed before, small errors incurred from rounding are not an issue when we are approximating the \emph{largest} eigenvalue. However, due to rounding when we are approximation the second-largest eigenvalue, our vector will not remain perpendicular to $v_n$. Hence, given enough iterations, the process will always converge to the largest eigenvalue. 
Naively, we can combat this by projecting the current vector $x$ on the subspace orthogonal to $v_n$ in every step, or in other words, by removing components in direction $v_n$. This  requires  computing the inner product between $v_n$ and $x$. The latter is no problem in the centralized setting, but computing this inner product requires $\Omega(D)$ rounds in \CONGEST. Hence we want to limit how often we do this. We show that the error blows up at a rate that means that we only need to project periodically, in total at most  $\lceil \tfrac{\lambda_2}{\eps}\rceil$ times. In order to perform this efficiently, we show that all interesting cases actually have a diameter small enough to repeatedly perform these projections without increasing the runtime too much.

\paragraph{Approximating $\lambda_1,\ldots,\lambda_k$.} 
We compute the eigenvalue approximations sequentially. For approximating $\lambda_i$, we use the same general approach as for $\lambda_2$. Instead of projecting on the subspace orthogonal the largest eigenvalue, we project on the subspace orthogonal to the largest $i-1$ eigenvalues, and remove the resulting parts from consideration. The key challenge is that we do not know the corresponding eigenvectors exactly but they are only approximately known from our prior steps. We  show that the power method still converges despite these inaccuracies in the eigenvectors that we project on. This is one of the most technical parts of this paper and appears in \Cref{sec:kth_ev}.

\paragraph{Probabilistic Guarantees.}
The power method as stated in \Cref{lm:powerMethod} is correct with constant probability. This means that we can for example can also obtain our sparsest cut result, \Cref{thm:CONGESTSparsestCut}, in $O(\log n/\phi)$ rounds but with \emph{constant} probability. To boost the probability, we run $O(\log n)$ independent instances of the power method and return the maximum. Since we are looking for the largest eigenvalue, and we know that we cannot overshoot (see \Cref{lm:no_overshoot}), this gives the required approximation with high probability.

\paragraph{Lower Bound Constructions.}
We construct graphs by combining well-understood components---such as cliques, paths, and star graphs—where the sparsest cut is easy to determine. Using these building blocks, we define two graph families with significantly different sparsest cut values, yet with locally identical neighborhoods around certain vertices. As a result, these vertices cannot distinguish between the two families and thus cannot accurately approximate the sparsest cut without global information. This implies that approximating the sparsest cut requires $\Omega(D)$ rounds in our constructions.

It is challenging to use a similar approach to obtain a lower bound for eigenvalue approximation, as determining the eigenvalues of any graph is highly non-trivial. The eigenvalues of few specific graphs have been computed, like paths, cycles, and cliques, but determining them remains hard in general. This stems partly from the fact that gluing together graphs can have big effects on the eigenvalues. Cauchy's interlacing eigenvalue theorem~ \cite{hwang2004cauchy} gives some bounds, but this is unfortunately of no help. Another approach is to use the fact that we know the sparsest cut, together with Cheeger's Inequality. Rephrased Cheeger's Inequality gives $\phi^2/2\le\lambda_2 \le 2\phi$. So it gives a linear upper bound on the eigenvalue given the sparsest cut, but it has a square in the lower bound. At the end of the day, this means that only using Cheeger's would result in a lower bound of the form $\Omega(\sqrt D)$. For our $\Omega(D)$ lower bound, we need a better upper bound on the eigenvalues of the graphs used in the lower bound construction. We exploit that eigenvalues can also be bounded by specific vectors. The simplest such statement is that the largest eigenvalue of a matrix $M$ equals $\max_{x\in \R^n} \tfrac{x^T Mx}{x^Tx}$. Hence any $x$ can provide a lower bound for the largest eigenvalue. We use a similar statement for the second smallest eigenvalue and construct an appropriate vector $x$ to prove a tighter upper bound on the eigenvalues of our constructed graphs. In total, we are able to provide upper and lower bounds that characterize the eigenvalues of the graphs used in our lower bound construction up to a constant factor. 

\paragraph{Sparsest Cuts and Cheeger's Inequality.}
Finally, let us explain how the eigenvalue results enable us to derive an approximation of the sparsest cut.

Let us obtain some intuition why eigenvalue $\lambda_2$ relates to the connectedness of the graph. In the extreme case that $\lambda_2=0$ holds the graph is disconnected:  On each connected component, the degree vector is still an eigenvector with eigenvalue $0$. Hence a vector that has $\sqrt{\deg}$ on one connected component, and all $0$s on the other connected component is also an eigenvector -- and linearly independent of $\sqrt{\deg}$. The reverse direction, i.e., that the graph is connected if $\lambda_2\neq 0$ holds, is also true but its proof is slightly more involved, so we omit it in this overview, see e.g.~\cite{chung1997spectral}. 

Cheeger's Inequality~\cite{AlonM85,Alon86,SinclairJ89} gives a more precise statement about the connectedness of the graph and the second smallest eigenvalue: it relates $\lambda_2$ to the sparsest cut~$\phi$. 

\begin{theorem}[Cheeger's Inequality]\label{thm:Cheeger}
    Let $G = (V, E)$ be an undirected, weighted graph, then
    \begin{equation*}
        \frac{\lambda_2}{2} \le \phi \le \sqrt{2 \lambda_2}~.
    \end{equation*}
\end{theorem}
The fact that $\lambda_2=0$ holds if and only if the graph is disconnected is a corollary of \Cref{thm:Cheeger}. 

Since Cheeger's Inequality states that the value of $\lambda_2$ approximates $\phi$, we can find an approximation $\tilde \lambda_2$ for $\lambda_2$ with \Cref{thm:EigenvaluesMain} and use it to output an approximation $\tilde \phi = \sqrt{2\tilde \lambda_2}$ of the sparsest cut, establishing \Cref{thm:CONGESTSparsestCut}. 

For larger values of $k$, similar statements hold. First, we remark that $\lambda_k=0$ if and only if the graph consists of at least $k$ connected components. One direction of this statement is again easy to see: as before, on each connected component the $\sqrt{\deg}$ vector is an eigenvector with eigenvalue $0$. In total, this gives $k$ independent eigenvalues with eigenvector $0$ -- in other words, $\lambda_k=0$. 

Similar as for $k=2$, this generalizes by relating $\lambda_k$ to $\phi_k$:
Lee, Gharan, and Trevisan showed the following theorem which we then use to establish \Cref{thm:CONGESTkWaySparsestCut} via \Cref{thm:EigenvaluesMain}.

\begin{theorem}[Higher-Order Cheeger's Inequality, \cite{LeeGT12}]\label{thm:HO_Cheeger}
Let $G = (V, E)$ be an undirected, weighted graph, then
 $   \frac{\lambda_k}{2}\le \phi_k \le O(k^2) \sqrt{\lambda_k}~.$
\end{theorem}

\subsection{Outline}
We first give our main results on eigenvalue estimation in \Cref{sec:EV}, i.e., we prove \Cref{thm:EigenvaluesMain}. This sections starts with the formal description of the Laplacian matrix and the power method, and then proves \Cref{thm:EigenvaluesMain} in reverse order: (3) in \Cref{sec:largest_ev}, (2) in \Cref{sec:smallest_ev}, and (1) in \Cref{sec:kth_ev}. We conclude \Cref{sec:EV}, with our lower bound on eigenvalue estimation, \Cref{thm:eigenvalue_LB}, and a description of how to use similar methods to approximate the entire spectrum.
Then, in \Cref{sec:sparsest_cut}, we provide our upper and lower bounds on ($k$-way) sparest cuts. 
To be self-contained, we give a proof for the power method in \cref{app:powerMethod}.

  \newpage
  \section{Approximating Eigenvalues} \label{sec:EV}
In \Cref{sec:Laplacian_preliminaries}, we begin with the necessary background and introducing notation. In \Cref{sec:power_method_description}, we give a formal description of the power method and state its most important properties. 
Then, we prove \Cref{thm:EigenvaluesMain} part (2) and part (3), \Cref{thm:CONGEST_largest,thm:CONGEST_smallest}, proceeding in the reverse order: First, we give the algorithm for the largest eigenvalue $\lambda_n$ in \Cref{sec:largest_ev}. Then, we give the algorithm for $\lambda_2$ in \Cref{sec:smallest_ev}.

\subsection{The Graph Laplacian Matrix.}\label{sec:Laplacian_preliminaries}
First, we will give some definitions and notation. For a more thorough introduction of spectral graph theory, we refer to, e.g., \cite{chung1997spectral}. 

The normalized Laplacian is defined as $L:=I- \Deg^{-1/2}A\Deg^{-1/2}$, where $\Deg$ is a diagonal matrix with the (weighted) vertex degrees on the diagonal, and $A$ is the adjacency matrix. Let us inspect some of its properties. First of all, we note it is \emph{symmetric}, i.e., $L^T=L$. We know that all symmetric matrices allow for a set of \emph{orthonormal} eigenvector $v_1, \dots, v_n$. That means the eigenvectors are \emph{orthogonal}: $\langle v_i,v_j\rangle=0$ for $i\neq j$, and have norm 1: $||v_i||=1$ for all $i$. Moreover, any Laplacian matrix is \emph{positive semi-definite}. This means that all its eigenvalues are nonnegative. Equivalently, we have for any $x\in \R^n$ that $x^TLx\ge 0$. More precisely, it has eigenvalues $0=\lambda_1\leq \lambda_2 \leq \dots \le \lambda_n\le 2$. The eigenvalue $0$ can easily be verified since $Lv_1=0$ for $v_1:=\sqrt{\deg}$, the vector with in each entry the square root of the degree of the corresponding vertex. 

As explained in \Cref{ssec:techOverview}, $\lambda_2=0$ if and only if the graph is disconnected. More generally, $\lambda_k=0$ if and only if the graph consists of $k$ connected components. In the \CONGEST model, we may assume a connected graph, or alternatively apply the results on each component of the communication network separately.  Hence we have $\lambda_2 >0$. 


\paragraph{Notation.}
Throughout, $x_1,x_2, \dots$ denote distinct vectors. We write $(x)_v$ for the $v$-th entry of a vector $x$, i.e., the entry of the vector $x$ that is hold by vertex~$v$. 

We reserve $0=\lambda_1\leq \lambda_2 \leq \dots \le \lambda_n\le 2$ for the eigenvalues of the normalized Laplacian. For a general $n$ by $n$ matrix $M$, we denote its eigenvalues by $\mu_1 \le \mu_2 \le \dots \le \mu_n$. We note that in the following $M$ and $L$ might differ, however, they will always have the same eigenvectors. Denote these as $v_1, \dots, v_n$, corresponding to $\mu_1, \dots, \mu_n$ and $\lambda_1, \dots, \lambda_n$ respectively. 

\subsection{Power Method}\label{sec:power_method_description}

The core of our eigenvalue approximation is an efficient CONGEST implementation of the power method. We recall its main measure of progress. 
\PowerMethod*

 Most results in this section have appeared in the literature before and we do not claim any novel contribution. To be self-contained and as our statements slightly differ from statements in the literature we present proofs of all presented lemmas in \Cref{app:powerMethod}. To the best of our knowledge, the only exception is the bit-precision lemma that we present at the end of the section. It is central to implementing the power method in CONGEST and essentially says that rounding intermediate values does not significantly affect the performance of the power method; also see the paragraph on `exploding values' in \Cref{ssec:techOverview}. It may be slightly stronger than what is needed in the centralized setting and potentially also of independent interest.

\Cref{lm:powerMethod} states that there is a constant probability that the power method fails to find the largest eigenvalue. Hence, we need to run $O(\log n)$ instances of it and return the largest eigenvalue found during any of the independent runs. The following lemma shows that we never return any value larger than the largest eigenvalue. It is also needed to control the growth of numbers which is essential to implement the method in CONGEST.
 
\begin{restatable}{lemma}{lemNoOvershoot}\label{lm:no_overshoot}
    For any vector $x\in \R^k$ we have that  $\frac{x^TMx}{x^Tx}\le \mu_n$.
\end{restatable}

Regardless of whether in the centralized setting or in CONGEST, the intermediate vectors may consist of quickly growing values that cannot be efficiently represented. The following lemma states that truncating all values after $O(\log n)$ bits does not affect the outcome of the power method by too much. Note that the statement may sound trivial as a precision of $1/\poly n$ for each single number seems to simply result in an imprecision of $k/\poly n$ after $k$ steps. A closer look at the power method reveals that errors may amplify and exponentially grow over time. Hence, the next lemma requires an actual proof that appears in \Cref{app:powerMethod}.
\begin{restatable}{lemma}{lemNumStab}\label{lm:num_stab}
    Let $k\ge i \ge 1$. Throughout the power method, rounding each element of $x_i:=M^ix_0$  to the $\Theta(\log n)$-th bit, does not impact the outcome of $\frac{(x_k)^TMx_k }{(x_k)^Tx_k}$ by more than a factor $(1\pm n^{-\Theta(1)})$.
\end{restatable}

\subsection{Warm-up: The Largest Eigenvalue}\label{sec:largest_ev}


First, we show how the power method applied directly to the Laplacian gives an approximation of the largest eigenvalue. 
\begin{restatable}[Approximating $\lambda_n$]{theorem}{CONGESTlargest}\label{thm:CONGEST_largest}
    Let $\eps>0$. There is a randomized \CONGEST algorithm that, given an undirected, weighted graph $G=(V,E)$, with high probability computes a $(1\pm \eps)$-approximation of the largest eigenvalue $\lambda_n$ of the normalized Laplacian. The algorithm takes $O( D+\log^2 n/\eps )$ rounds. 
\end{restatable}
\begin{proof}
    We use the power method, \Cref{lm:powerMethod}, with $M=L$ and $k=\Theta(\log n/\eps)$ iterations. Each iteration consists of a matrix-vector multiplication, in which each vertex $v$ needs to send its current value of $(x)_v$ to its neighbors.    
    After this, we perform one additional matrix-vector multiplication and compute $\tfrac{x^TMx}{x^Tx}$ in $O(D)$ rounds.


    \textbf{Bit precision.}
    \Cref{lm:num_stab} states that we can truncate each value after the $O(\log n)$ most significant bits. We  additionally remember the order of magnitude of the values. Since we have $k=\Theta(\log n/\eps)$ iterations, and our numbers are of size $2^k$, we need at most $\log\log (2^k)=\log(\log n/\eps) =O(\log n)$ bits to store the order of magnitude. Here we use that without loss of generality $\eps =\Omega(1/n^2)$. 
    We conclude that we can communicate each value up to the required precision in $O(\log n)$ bits. 

    \textbf{Success Probability.} 
    Note that \Cref{lm:powerMethod} only gives constant probability where the probability is over the choice of the random start vector. Repeating the process $O(\log n)$ times and take the maximum to obtain the result with high probability, since we cannot overshoot by \Cref{lm:no_overshoot}. 
    The matrix-vector multiplications take $O(\log n \cdot \frac{\log n}{\eps})$. However, for every instance, we also need to compute a final inner product that takes $O(D)$ time each. 
    We can pipeline these inner product for the $O(\log n)$ instances to obtain congestion plus dilation time for this part: $O(\log n+D)$. Combining the running times, we obtain 
    $O(\log^2 n/\eps+ D+\log n)=(\log^2 n/\eps+ D )$ rounds. 
\end{proof}

  \subsection{The Smallest Non-Zero Eigenvalue}\label{sec:smallest_ev}
Next, we consider the smallest non-zero eigenvalue. The main idea is to consider $M=2I-L$, such that the smallest eigenvalue becomes the largest, and then apply the power method. However, we are looking at the \emph{smallest non-zero} eigenvalue. So we take additional measures to avoid finding the smallest ($=0$) eigenvalue. 



\begin{restatable}[Approximating $\lambda_2$]{theorem}{CONGESTsmallest}\label{thm:CONGEST_smallest}
Let $\eps>0$. There is a randomized \CONGEST algorithm that, given an undirected, weighted graph $G=(V,E)$, with high probability computes an approximation of the smallest non-zero eigenvalue $\lambda_2$ of the normalized Laplacian. The algorithm takes 
\begin{equation*}
    O\left(\frac{\log^2 n}{\eps} + \frac{\log n}{\eps}\cdot \left\lceil \frac{\lambda_2}{ \eps}\right\rceil\right) = O\left(\frac{\log^2 n}{\eps} + \frac{\log n}{\eps^2}\right)
\end{equation*} rounds. At the end of the algorithm, every vertex outputs an estimation $ \lambda_2 \le \tilde \lambda_2 \le \lambda_2 +\eps$.
\end{restatable}
\begin{proof}
     We check if $D>\Theta(\log n/\eps)$ holds in $O(\log n/\eps)$ rounds. If so, then we have $\log n/\eps< \Theta(D) \le \Theta(\log n/\phi)$ by \Cref{lm:diam_k_sparsest_cut}, or equivalently $\phi< \Theta(\eps)$. By Cheeger's Inequality, \Cref{thm:Cheeger}, $\lambda_2\le 2 \phi \le 2\eps$. So we simply output $\tilde \lambda_2=\Theta(\eps)$ and we are done. 

    So for the remainder of the proof, we may assume $D= O(\log n/\eps)$. The algorithm is similar as in \Cref{thm:CONGEST_largest}, however, we now need to find the \emph{smallest} (nonzero) eigenvalue. To flip smallest and largest, we use a standard trick and consider the matrix $M=2I- L$ (instead of $M=L$ as in \Cref{thm:CONGEST_largest}). $M$ is still a PSD matrix, now with largest eigenvalue $2$, for which we also know the largest eigenvector $v_n := \sqrt{\deg}$. We denote the eigenvalues of $M$ by $0\le \mu_1\le \dots \le \mu_n=2$.
    If we were to apply the power method on $M$, we would find this largest eigenvalue $\mu_n=2$. Our goal is, however, to find $\mu_{n-1}$. Thus,  we aim to use it on the subspace orthogonal to $v_n$ by nullifying contributions in direction $v_n$. Formally, we change the matrix $M$ yet again.  
    
     To this end consider $M':=M-2v_n v_n^T$. All eigenvectors and most eigenvalues  stay the same: for $i<n$ we have 
    \begin{equation*}
        (M-2v_n v_n^T)v_i= Mv_i -2v_n\langle v_n,v_i\rangle=Mv_i=\mu_i v_i,
    \end{equation*} 
    by orthogonality of the eigenvectors. Only eigenvector $v_n$ changes its eigenvalues from $2$ in matrix $M$ to  $0$ in matrix $M'$: 
    \begin{equation*}
        (M-2v_n v_n^T)v_n= Mv_n -2v_n\langle v_n,v_n\rangle=Mv_i=2 v_n-2v_n=0.
    \end{equation*} 
    This means that the largest eigenvalue of $M'$ is $\mu_{n-1}$ Suppose now we have found $\tilde \mu_{n-1}$ such that $(1-\eps)\mu_{n-1} \le \tilde \mu_{n-1} \le \mu_{n-1}$ holds. Then also $\mu_{n-1} = 2- \lambda_2$ holds, so we output $\tilde \lambda_2 := 2-\tilde \mu_{n-1}$.  We obtain
    \begin{align*}
        \tilde\lambda_2 &= 2-\tilde \mu_{n-1} \ge  2- \mu_{n-1} = \lambda_2; \text{and} \\
        \tilde\lambda_2 &= 2-\tilde \mu_{n-1} \le 2- (1-\eps) \mu_{n-1} = 2- (1-\eps)(2-\lambda_2) = (1-\eps)\lambda_2 +2\eps \le \lambda_2+2\eps. 
    \end{align*}
    Setting $\eps\leftarrow\eps/2$ gives the result. 

    \textbf{Finding the largest eigenvalue of $M'$.}
    Again, we use the power method, \Cref{lm:powerMethod}, to approximate $\mu_{n-1}$.
 We use the next claim to prove that for most steps, multiplying with $M'$ is actually the same as multiplying with $M$.
\begin{claim} 
\label{claim:MprimeMultiplication}
For all $x$, we have 
\begin{enumerate}
\item $M'x= Mx-2v_n v_n^Tx = Mx-2v_n\langle x,v_n\rangle$;
\item $\langle M'x,v_n \rangle=0$.
\end{enumerate}
\end{claim}
\renewcommand{\qed}{\ensuremath{\hfill\vartriangleleft}}
\begin{proof}
The first part of the claim is immediate. For the second part we compute
\begin{align*}
      \langle M'x,v_n \rangle  & = \langle (M-2v_n v_n^T)x,v_n \rangle = \left\langle (M-2v_n v_n^T)\sum_{i=1}^n\langle x,v_i\rangle v_i,v_n \right\rangle\\
        &= \left\langle \sum_{i=1}^n\mu_i\langle x,v_i\rangle v_i-2v_n\langle x,v_n\rangle,v_n \right\rangle = \sum_{i=1}^n\mu_i\langle x,v_i\rangle \langle v_i,v_n\rangle -2\langle x,v_n\rangle \langle v_n,v_n \rangle \\
        &= 2\langle x,v_n\rangle-2\langle x,v_n\rangle =0. & \qedhere
    \end{align*}
\end{proof}

By \Cref{claim:MprimeMultiplication} part 1, we can compute $x_1=M'x_0$ by a multiplication with $M$ and inner product between $x_0$ and $v_n$. The latter can be computed in $O(D)$ time. Due to \Cref{claim:MprimeMultiplication} part 2, we have that the inner product between $x_i$ and $v_n$ is zero for all $i\geq 1$. Hence, after the first iteration we obtain $M'x_i=Mx_i$ and we can simply multiply with $M$ instead of $M'$.

    \textbf{Bit precision.}
    As in \Cref{thm:CONGEST_largest}, we have by \Cref{lm:num_stab} that $O(\log n)$ bits suffice to store each value. 

    However, rounding might introduce a small error in the direction of $v_n$. As opposed to \Cref{lm:num_stab}, this error could blow up over time: it scales with $\mu_n^l$ after $l$ rounds, while the vector's component we want scales with $\mu_{n-1}^l$. The number of iterations $l$ after which this becomes a problem that depends on the ratio $\mu_n/\mu_{n-1}$. We will let this error blow up for $l$ iterations, and then project again on the subspace orthogonal to $v_n$, using diameter time as above. Before we describe that, we prove how large $l$ can be. 
    
    Formally, we prove the following claim, where $\delta v_n$ denotes the error. For errors perpendicular to $v_n$, we can refer to \Cref{lm:num_stab}, which also shows rounding errors are of size at most $\delta\le n^{-\Theta(1)}||x_i||$. 
    \begin{claim}
        Let $x_i$ be a vector such that $\langle x_i, v_n\rangle=0$, and $0<\delta \le n^{-\Theta(1)}||x_i||$ be a parameter. Then for $l \le O(\tfrac{\log n}{\lambda_2})$ we have
        \begin{equation*}
            \left| \frac{(x_i+\delta v_n))^T M^{2l+1} (x_i+\delta v_n)}{(x_i+\delta v_n)^T M^{2l}(x_i+\delta v_n)} - \frac{x_i^T M^{2l+1} x_i}{x_i^T M^{2l} x_i}\right| \le  n^{-\Theta(1)}\frac{x_i^T M^{2l+1} x_i}{x_i^T M^{2l} x_i}.
        \end{equation*}
    \end{claim}
    \begin{proof}
        We first investigate the first term, i.e., the Rayleigh coefficient including the error after $l$ iterations 
    \begin{align*}
        \frac{(x_i+\delta v_n))^T M^{2l+1} (x_i+\delta v_n)}{(x_i+\delta v_n)^T M^{2l}(x_i+\delta v_n)} &= \frac{x_i^T M^{2l+1} x_i+ \delta^2\mu_n^{2l+1}}{x_i^T M^{2l} x_i+ \delta^2\mu_n^{2l}} \le \frac{x_i^T M^{2l+1} x_i+ n^{-\Theta(1)}||x_i||^2\mu_n^{2l+1}}{x_i^T M^{2l} x_i- n^{-\Theta(1)}||x_i||^2\mu_n^{2l}},
    \end{align*}
    where the equality follows from orthonormality of the eigenvectors, and the inequality follows from our bound on the error. If we now show for the numerator that $n^{-\Theta(1)}||x_i||^2\mu_n^{2l+1} \le n^{-\Theta(1)} x_i^T M^{2l+1} x_i$, then this also holds for the denominator (since this is a factor of $M$ less), and hence we obtain that the Rayleigh coefficient after $l$ iterations changed by a factor at most $(1+n^{-\Theta(1)})$. Hereto, we inspect $x_i^T M^{2l+1} x_i$. We note that $|\langle x_0, \rangle| \ge \tfrac{3}{16}$. We obtain
    \begin{align*}
        x_0^T M^{2(i+l)+1} x_0= x_0^T M^{2(i+l)+1} x_0 \ge \mu_{n-1}^{2(i+l)+1} |\langle x_0, \rangle|^2 \ge \mu_{n-1}^{2(i+l)+1} \left(\frac{3}{16}\right)^2. 
    \end{align*}

    On the other hand, we have that 
      $ ||x_i||^2 = x_0^T M^{2i} x_0 \le  \mu_{n-1}^{2i}\cdot ||x_0||^2 =\mu_{n-1}^{2i}\cdot n. $
    Putting this together, we obtain that 
    \begin{align*}
        n^{-\Theta(1)}||x_i||^2\mu_n^{2l+1} &\le n^{-\Theta(1)}\cdot \mu_{n-1}^{2i} n\cdot \mu_n^{2l+1} = n^{-\Theta(1)}\cdot \mu_{n-1}^{2(i+l)+1}\cdot \left(\frac{\mu_n}{\mu_{n-1}}\right)^{2l+1}\\
        &\le n^{-\Theta(1)}\cdot\left(\frac{\mu_n}{\mu_{n-1}}\right)^{2l+1}\cdot  x_0^T M^{2(i+l)+1} x_0.
    \end{align*}
    If we now show that 
       $ \left(\frac{\mu_n}{\mu_{n-1}}\right)^{2l+1} \le n^{\Theta(1)},$
    then we are done. Here we use that $l$ is bounded: after many iterations, the largest eigenvector will always dominate. By assumption, we have $l = \Theta(\tfrac{\log n}{\lambda_2})$. Using this, together with $\mu_n=2$ and $\mu_{n-1}=2-\lambda_2$, we see
    \begin{align*}
        \left( \frac{\mu_n}{\mu_{n-1}}\right)^l &= \left( \frac{2}{2-\lambda_2}\right)^l \le (1+\lambda_2)^l \le \exp(\lambda_2\cdot l)=  \exp(\Theta(\log n)=n^{\Theta(1)}. & \qedhere
    \end{align*}
    \end{proof}
    \renewcommand{\qed}{\hfill \ensuremath{\blacktriangleleft}}

    This claim shows that for $l=\Theta(\log n /\lambda_2)$ iterations, the rounding error does not propagate too much.     
    After every $\Theta(\log n /\lambda_2)$ iterations, we orthogonalize again. We can simply do this by computing $x-\langle x,v_n\rangle v_n$, since $v_n$ is known: $v_n=\sqrt{\deg}$. However, this leads to $O(\left\lceil \tfrac{\lambda_2}{ \eps}\right\rceil)$ additional inner products.


    \textbf{Success Probability.}
    Again, we note that \Cref{lm:powerMethod} only gives constant probability. We can simply repeat the process $O(\log n)$ times and take the maximum to obtain the result with high probability, since we cannot overshoot by \Cref{lm:no_overshoot}. As before, we can stream the $O(\lceil \lambda_2/\eps \rceil)$ inner products, so the total time becomes: 
    \begin{equation*}
         O\left(\frac{\log^2 n}{\eps}+\log n\left\lceil \frac{\lambda_2}{ \eps}\right\rceil + D\left\lceil \frac{\lambda_2}{ \eps}\right\rceil\right) =   O\left(\frac{\log^2 n}{\eps} + \frac{\log n}{\eps}\cdot \left\lceil \frac{\lambda_2}{ \eps}\right\rceil\right). 
    \end{equation*}

    We can also simplify this to 
        $O\left(\frac{\log^2 n}{\eps} + \frac{\log n}{\eps^2}\right)$~,
    using that $\lambda_2\le 2$.  
\end{proof}

  \subsection{The \texorpdfstring{$k$}{k} Smallest Eigenvalues}\label{sec:kth_ev}
In this section, we generalize the method from \Cref{sec:smallest_ev}. The general idea remains to change the matrix by projecting on the space orthogonal to the found eigenvectors. However, the found eigenvectors are now no longer exact, but approximate. We first show that the power method can work with such approximations as well. 
The following lemma captures the main result of this section.

\begin{lemma}\label{lm:power_method_for_k}
Let $k\ge 0$ be a constant. 
Suppose $\tilde \mu_i,\tilde v_i$ for $i\ge n-k$ are orthonormal vectors that approximate the eigenvalues thus-far: $\tilde \mu_i=\tilde v_i^T M \tilde v_i $ and $(1-\poly(\eps) )\mu_i \le \tilde \mu_i \le \mu_i$. 
    Performing $O(\log n)$ instances of the power method -- each with $O(\log n/\eps)$ iterations -- with the matrix 
    $M_{\rm{approx}}:=M- \sum_{i>n-k}\tilde \mu_i \tilde v_i \tilde v_i^T$ gives w.h.p.\ either a $(1-\eps)$-approximation $\tilde \mu_{n-k}$ of the $k$-th largest eigenvalue $\mu_{n-k}$, together with a corresponding vector $\tilde v_{n-k}$, or we can output that $\mu_{n-k}<\eps$. 
\end{lemma}

Suppose we know the exact eigenvalues and eigenvectors $\mu_i,v_i$ for $i>n-k$, consider $M_{\rm{ideal}}:=M- \sum_{i>n-k} \mu_i  v_i v_i^T$. This matrix has the same eigenvectors $v_i$, but now largest eigenvalue $\mu_{n-k}$. This follows from orthonormality of the eigenvalues. Indeed we see for $i>n-k$
\begin{align*}
    (M- \sum_{i>n-k} \mu_i  v_i v_i^T)v_j &= Mv_j -\mu_jv_j =0\\
    \intertext{and for $i\le n-k$ }
    (M- \sum_{i>n-k} \mu_i  v_i v_i^T)v_j &= Mv_j = \mu_j v_j.
\end{align*}
So if we obtain a $(1-\eps)$-approximation of the largest eigenvalue of matrix $M_{\rm{ideal}}$, we are done. We will show that the largest eigenvalue of $M_{\rm{approx}}$ is a $(1-\eps)$ approximation of the largest eigenvalue of $M_{\rm{ideal}}$. The power method gives with constant probability a $(1-\eps)$-approximation of the largest eigenvalue. Repeating this $O(\log n)$ times and taking the maximum over these iterations gives the result with high probability. Note that we cannot overshoot by \Cref{lm:no_overshoot}, hence taking the maximum indeed works. We can stream these iterations such that the $\log n$-factor does not multiply with $D$. 

We note that a big difference to \Cref{thm:CONGEST_smallest} is that we now do not have exact orthogonality \emph{at any point} in time. This also means that rounding does not change this fact. Such errors are corrected by taking inner products out \emph{every step} -- rather than periodically. 

Furthermore, the reader might notice that we state \Cref{lm:powerMethod} for positive semi-definite matrices, and we no longer claim that $M_{\rm{approx}}$ is positive semi-definite. There are two solutions for this: 1) the power method also works on general matrices~\cite{golub2013matrix}. The running time is a constant factor slower, and it finds the \emph{dominant} eigenvalue rather than the largest. 2) We note that $M_{\rm{ideal}}$ \emph{is} positive semi-definite. Since we show that the eigenvalues of $M_{\rm{ideal}}$ and $M_{\rm{approx}}$ are close, this means that the eigenvalues in $M_{\rm{approx}}$ are at least $-\eps$. By running the power method on $M_{\rm{approx}}+\eps I$, we have a PSD matrix again, and can use \Cref{lm:powerMethod} as before. The additive matrix $+\eps I$ has  a linear effect on the eigenvalue approximation, so we can simply remove it at the end.

It remains to show that $M_{\rm{approx}}$ approximates $M_{\rm{ideal}}$, which we do in \Cref{sec:Mapprox_Mideal}. Then we show how this applies to the Laplacian in the \CONGEST model in \Cref{sec:CONGEST_lambda_k}.

\subsubsection{\texorpdfstring{$M_{\rm{approx}}$}{Mapprox} approximates \texorpdfstring{$M_{\rm{ideal}}$}{Mideal}.}\label{sec:Mapprox_Mideal}
First, we note that we can assume the spectrum to be relatively simple. 
\begin{lemma}\label{lm:spectrum_gaps}   
    Without loss of generality, we can assume that the eigenvalues take values in $\{\mu_n, (1-\eps)\mu_n, (1-\eps)^2\mu_n, \dots\}$. 
\end{lemma}
\begin{proof}
    Instead of analyzing the matrix $M$, with orthonormal eigensystem $\{(\mu_i,v_i)\}$, we analyze the following system. We round every eigenvalue $\mu_i$ down to $\hat \mu_i$, the next value in $\{\mu_n, (1-\eps)\mu_n, (1-\eps)^2\mu_n, \dots\}$. Then we consider the matrix 
    \begin{equation*}
        M' := M - \sum_{i\in[n]}(\mu_i-\hat \mu_i)v_iv_i^T.
    \end{equation*}
    By orthonormality of the eigenvectors, $M'$ has the same eigenvectors. Its eigenvalues, however, have changed to $\hat \mu_i$. This gives an error of at most $\epsilon \mu_i$ for each eigenvalue $\mu_i$. Setting $\eps\leftarrow\eps/2$, makes that this additional error does not impact the final answer. 
\end{proof}

In particular, this means that the second largest eigenvalue is at most $(1-\eps)\mu_n$. This gives us the following lemma, stating that the previous iterations are well aligned with the eigenvectors. 

\begin{lemma} \label{lm:high_inner_prod}
    Suppose $||y||=1$ and $\frac{y^TMy}{y^Ty}\ge (1-\delta)\mu_{n}$. Then 
    \begin{equation*}
        \langle y,v_n\rangle \ge \sqrt{1-\delta/\epsilon}. 
    \end{equation*}
\end{lemma}
\begin{proof}
    Since $|y||=1$, we have $y^TMy \ge (1-\delta)\mu_n$. We use this to see: 
    \begin{align*}
        \sum_{i=1}^n \langle y,v_i \rangle^2 \mu_i  &= \left(\sum_{i=1}^n \langle y,v_i \rangle v_i\right)^T M \left(\sum_{i=1}^n \langle y,v_i \rangle v_i\right)\\
        &= y^TMy \ge (1-\delta) \mu_n.
    \intertext{Since the eigenvalues are $(1-\eps)$ apart, let $j$ the highest index such that $\mu_j<(1-\eps)\mu_n$, which means that $\mu_i=\mu_n$ for $i> j$. So}
         \sum_{i>j}\langle y,v_n\rangle^2 \mu_n + \sum_{i\le j}\langle y,v_n\rangle^2 (1-\eps)\mu_n &\ge (1-\delta)\mu_n. 
    \intertext{Using that $y$ has norm 1, and that $\mu_n>0$}
 \sum_{i>j}\langle y,v_n\rangle^2  + (1-\sum_{i>j}\langle y,v_n\rangle^2) (1-\eps) &\ge (1-\delta). 
    \intertext{Rearranging gives}
        \sum_{i>j} \langle y,v_i\rangle^2 &\ge 1-\delta/\epsilon.
    \end{align*}
    However, $v_i$ for $i>j$ all have the same eigenvalue $\mu_n$, so we can just rearrange the orthonormal basis such that $\langle y,v_n\rangle \ge \sqrt{1-\delta/\epsilon}$.
\end{proof}

Intuitively, this means that if we  set $\delta=\eps^2$, we get $\langle y,v_n\rangle \ge \sqrt{1-\eps}\ge 1-\eps$. In reality, the analysis is a bit more involved -- so we carry the $\delta$ for now.



Next, we show that if the inner product of two vectors is small, the difference in their projections is also small. Here, we denote $||A||$ for the \emph{matrix norm} or \emph{operator norm} of a matrix $A$, defined by $||A||:=\max_{||x||=1} ||Ax||=\max_{||x||=1} x^T A^TAx$. Note that in our case of symmetric, positive semi-definite matrices, the norm of a matrix is bounded by its largest eigenvalue squared.

\begin{lemma}\label{lm:projections}
    If $||v||=||w||=1$, $\langle v,w\rangle \ge 1-\eta$ and $a,b\ge 0$, $0<\eta<1$ such that $(1-\eta) a\le b\le a$, then we have 
    $|| a vv^T-bww^T|| \le 5\eta b$.
\end{lemma}
\begin{proof}
First we note that
    \begin{align*}
        || a vv^T-bww^T|| &= ||(a-b)vv^T +b(vv^T-ww^T)|| \le ||(a-b)vv^T||+||b(vv^T-ww^T)||\\
        &=|a-b|\cdot ||vv^T||+b||vv^T-ww^T||\le b(\eta +||vv^T-ww^T||).
    \intertext{Next, we bound $||vv^T-ww^T||$:}
        ||vv^T-ww^T||&= ||(v-w)v^T+w(v-w)^T|| \le ||(v-w)v^T||+||w(v-w)^T|| \\
        &= ||v-w||\cdot||v||+||w||\cdot||v-w||= 2||v-w||, 
    \intertext{where we know that} 
        ||v-w||&=\langle v-w,v-w\rangle= \langle v,v\rangle +\langle w,w\rangle -2\langle v,w\rangle =2(1-\langle v,w\rangle).
    \end{align*}

    Putting this together we obtain 
    \begin{align*}
        || a vv^T-bww^T|| \le b(\eta +||vv^T-ww^T||) \le b(\eta +2||v-w||) \le b(\eta +4(1-\langle v,w\rangle)) \le 5\eta b. 
    \end{align*}
\end{proof}

Finally, we use the following theorem from perturbation theory (see, e.g., \cite{horn2012matrix}).
\begin{theorem}[Weyl's Inequality]\label{thm:weyl}
Let $A,E$ be real, symmetric $n$ by $n$ matrices. Let $\mu_i(A),\mu_i(A+E)$ denote the $i$-th eigenvalue of $A$, $A+E$ respectively, for $i=1, \dots n$. Then we have
    \begin{equation*}
        |\mu_i(A+E)-\mu_i(A)| \le ||E||.
    \end{equation*}
\end{theorem}

We apply \Cref{lm:high_inner_prod}, \Cref{lm:projections} and \Cref{thm:weyl} to our case:
\begin{align*}
    |\mu_i(M_{\rm{ideal}})-\mu_i(M_{\rm{approx}})| &\le ||M_{\rm{ideal}}-M_{\rm{approx}}||\\
    &= \left|\left| \sum_{i>n-k} \tilde \mu_i \tilde v_i \tilde v_i^T-\sum_{i>n-k} \mu_i  v_i v_i^T\right|\right|\\
    &\le  \sum_{i>n-k} \left|\left|\tilde \mu_i \tilde v_i \tilde v_i^T- \mu_i  v_i v_i^T\right|\right|\\
    &\le  \sum_{i>n-k} 5\cdot \delta/\eps \cdot \mu_i\\
    &\le  k\cdot 5\cdot \delta/\eps \cdot \mu_n.
\end{align*}

We need this to be at most $\eps \mu_{n-k}/2$, where the factor two comes from \Cref{lm:spectrum_gaps}. Hereto, we set $\delta_i := \tfrac{\eps^2}{10}\tfrac{\mu_n}{\mu_{n-k}}$. Finally,  we use that $\mu_n\le 2$ and that we can assume $\mu_{n-1}\ge \eps$. The latter holds, because if $\mu_{n-1}< \eps$, we still have that $\tilde \mu_{n-1} \le \mu_{n-1}$. It might be below $(1-\eps)\mu_{n-1}$, but we can recognize this case, as needed for \Cref{lm:power_method_for_k}.
So we can bound $\mu_n/\mu_{n-k}$ by $2/\eps$. To obtain our result, \Cref{lm:power_method_for_k}, it hence suffices to set $\delta_i = \tfrac{\eps^3}{20k}$.

\subsubsection{Approximating the \texorpdfstring{$k$}{k}-th Smallest Eigenvalue in \CONGEST }\label{sec:CONGEST_lambda_k}

\begin{restatable}[Approximating $\lambda_1,\ldots,\lambda_k$]{theorem}{CONGESTkth}\label{thm:congest_lambda_k}
    Let $\eps>0$ be a parameter and $k\ge1$ be a constant. There exists a $\CONGEST$ algorithm that, given an undirected, weighted graph $G=(V,E)$, with high probability finds $k$ approximations $\tilde\lambda_1,\tilde \lambda_2, \tilde \lambda_3,\dots, \tilde \lambda_k$ of the first $k$ eigenvalues of the normalized Laplacian such that 
    \begin{equation*}
       \lambda_i \le \tilde \lambda_i \le \lambda_i+\eps,
    \end{equation*}
    in $O(\log^2 n\poly(\eps^{-1}))$.
\end{restatable}
\begin{proof}
    First, we show that we can assume the diameter to be small. We check in $O(k\log n/\eps)$ rounds if the diameter is bigger than $\Theta(k\log n/\eps)$.
    As seen in \Cref{lm:diam_k_sparsest_cut}, if $D>\Theta(k\log n/\eps)$, we have that $\phi_k \le \eps$. 
    By \Cref{thm:HO_Cheeger}, we have that $\lambda_k \le 2\phi_k\le 2\eps$. So we set $\eps \leftarrow \eps/2$. And can then output $\tilde \lambda_i = \eps$ for all $i\le k$. For the remainder, we can assume that $D= O( k\log n/\eps)$. 

    By \Cref{lm:power_method_for_k}, we need to run the power method with the matrix $M_{\rm{approx}}=M- \sum_{i>n-k}\tilde \mu_i \tilde v_i \tilde v_i^T$. We need to run $O(\log n)$ instances to obtain our result with high probability, and each instance needs $O(\log n/\delta)$ iterations, with $\delta = \poly(1/\eps)$. Note that each iteration requires the previous iteration to be more precise. In the end reaching a bigger $\poly(\eps^{-1})$. 

    Note that we are not running the power method with a locally defined matrix anymore. We use $M_{\rm{approx}}:=M- \sum_{i>n-k}\tilde \mu_i \tilde v_i \tilde v_i^T$, where $M=2I-L$ \emph{is} locally defined. Compute $M_{\rm{approx}}x$ consists of multiple steps. First we compute $Mx=(2I-L)x$, which can be done in $O(1)$ rounds. Next, we need to compute each of the at most $k$ inner products $\tilde v_i^T x$. We can pipeline these inner products for the $\log n$-instances in $O(D+\log n)$ times. Recall that $\lambda_1=0$, with eigenvector $\sqrt{\deg}$. For $M_{\rm{approx}}$, this means the largest eigenvalue $\mu_n=2$ with eigenvector $v_n = \sqrt{\deg}$. We know this vector exactly, and do not need to use an approximation. 

    So we get total time $O(\log n\cdot \log n/\delta+\log n/\delta\cdot k \cdot (D+\log n)=O(\log^2n/\delta +\log nD)$.
    Filling in $D= O(k\log n/\eps)$ and $\delta = \poly(1/\eps)$ gives the result.
\end{proof}

\subsection{Lower Bounds For Eigenvalue Approximation}\label{sec:ev_LB}
For our lower bounds, we always assume that the graph is connected, as in standard in the \LOCAL and \CONGEST model. We allow for the general case that vertices know the value $n$.





\EigenvalueLB*
\begin{proof}
    We consider two classes of graphs.
    If with probability $>\tfrac{1}{2}$ it is correct on each, then there has to exist output that is correct on both. We consider this case. 
    
    First, consider a cycle on $k$ vertices $C_k$, connect this to a clique on $l=\Omega(n)$ vertices $K_l$. We denote this graph by $G_{k,l}$. 
    
    Second, consider a cycle on $k$ vertices $C_k$, connect this to a clique on $l$ vertices $K_l$ on one side, and a clique on $l'$ vertices. We denote this graph by $G_{k,l,l'}$. 

    \textbf{Claim 1.} We have $\lambda_2(G_{k,l})=\Theta(1/k^2)$. We consider the following eigenvector of the cycle $C_k$: $v_i = \sin(\tfrac{2\pi i}{k})$, which has eigenvalue $1-\cos(\tfrac{2 \pi}{k})$~\cite{chung1997spectral}. We note in particular that $v_0=0$. Since we connect the cycle to the clique at $v_0$, we can extent this vector to a vector that is all $0$s at other places. Now we have 
    \begin{equation}
        \frac{v^T L(G_{k,l})v}{v^Tv}=\frac{v^T L(C_{k})v}{v^Tv} = 1-\cos(\tfrac{2 \pi}{k}) = O(1/k^2).
    \end{equation}
    Further, by orthogonality we have that $v \perp w$, where $w=\sqrt{\deg}$ is the eigenvector corresponding to eigenvalue 0.
    One can show that $\lambda_2 = \min_{v\perp w}\frac{v^T L(G_{k,l})v}{v^Tv}$~\cite{chung1997spectral}. This means that for \emph{any} $v \perp w$ we have 
    \begin{equation*}
        \lambda_2 \le \frac{v^T L(G_{k,l})v}{v^Tv}.
    \end{equation*}
    We have $\phi(G_{k,l})=\Theta(1/k)$, there are two types of cuts to inspect to proof this: cutting through the cycle or cutting through the clique. Cutting through the cycle is most balanced, hence sparsest, if the cut $S$ is closest to the cycle. In this case we have    
     \begin{equation*}
        \phi(S) = \frac{|E(S,V\setminus S)|}{\min\{\Vol(S),\Vol(V\setminus S)\}}= \frac{2}{\min\{\Theta(l^2),2(k-1)\}} = \Theta(1/k).
    \end{equation*}
    
     Cutting through the clique always gives a bigger value: say we have $l'<l$ vertices of the clique in $S$, then 
     \begin{equation*}
        \phi(S) = \frac{|E(S,V\setminus S)|}{\min\{\Vol(S),\Vol(V\setminus S)\}}= \frac{l'(l-l')}{\min\{l'(l-1),(l-l')(l-1)+2k\}} \ge \frac{l-l'}{l-1} \ge \frac{1}{l-1}=\Theta(1/n).
    \end{equation*}    
    Since Cheeger's Inequality, \Cref{thm:Cheeger}, gives us that $\lambda_2 \ge \phi^2/2$, we get $\lambda_2(G_{k,l})=\Theta(1/k^2)$.

    \textbf{Claim 2.} We have $\lambda_2(G_{k,l,l'})= O(1/k^2)$. 
    The upper bound follows from the same argument as in Claim 1, with the additional observation that $v_{k/2}=0$ as well. 

    Suppose a vertex $v$ sees its $r$-hop neighborhood, for some $r< k/2$. This means the vertex sees at most one clique, say of $l$ vertices. The graph can still be $G_{k,l}$ or $G_{k,l,l'}$. For any $k\ge 2r$ (and $k\le n-l$). So we need to output $\tilde \lambda_2 \ge \Theta(1/r^2)$. However, the graph can also be $G_{2D,l,l'}$ with $\lambda_2(G_{2D,l,l'}) \le O(1/D^2)$. So we need $\Theta(1/r^2) \le (1+\eps)O(1/D^2)$, so $r= \Omega(D)$. 
\end{proof}

\subsection{Approximating the Entire Spectrum}
\label{sec:approx_spectrum}
Via similar methods, one can also approximate the entire spectrum. However, the output will not reveal how many eigenvalues there are of a certain approximation, i.e., it gives the spectrum \emph{without multiplicity}. 

\begin{remark}\label{remark:completeSpectrum}
    We can compute a multiplicative approximation of  the entire spectrum (without multiplicity) in $n^{o(1)}(\sqrt n+D)/\eps^2$ rounds.
\end{remark}
\begin{proof}[Sketch.]
Note that this statement does not have a linear dependency on $n$, even though there are $n$ eigenvalues and from a high level point of view we would need to address them sequentially with the power method. The linear dependency does not appear in the runtime as we only compute the spectrum approximately and without multiplicity. Thus, essentially we only need to test $O(\log (Wn)/\eps)$ regimes in $[0,2]$ for the existence of an eigenvalue. 

    The algorithm is now as follows. One can also find an eigenvalue with $\sigma$ by the power method on $(L-\sigma I)^{-1}$, with $\tilde \sigma$ some estimate of $\sigma$. By running this with  $\sigma= 2/(1+\eps)^i$ for $i=1,\cdots, O(\log(Wn)/\eps)$, we cover the entire spectrum. Now we have to perform matrix-vector multiplication with $(L-\sigma I)^{-1}$, which takes $n^{o(1)}(\sqrt n+D)$ rounds in each step~\cite{ForsterGLPSY21}. Multiplying this by $O(\log n/\eps)$ iterations, we obtain in total $n^{o(1)}(\sqrt n+D)/\eps^2$ rounds.
\end{proof}


    \newpage
    \section{Approximating the Sparsest Cut}\label{sec:sparsest_cut}
In this section, we use the approximate eigenvalues from \Cref{sec:approx_spectrum} to approximate the sparsest cut in the \CONGEST model. We also give the complexity of the problem in the \LOCAL model and we provide a lower bound in \emph{unweighted} graphs. 

\paragraph{\texorpdfstring{$k$}{k}-Way Sparsest Cut and Diameter}
We first observe that the diameter of a graph is related to its $k$-way sparsest cut. The result follows from a ball growing argument, that is (at least) known for the case $k=2$. We prove the generalization for completeness. 
\begin{lemma}\label{lm:diam_k_sparsest_cut}
    If $G=(V,E)$ has $k$-way sparsest cut $\phi_k$, then we have $D=O(k \log n/\phi_k )$. 
\end{lemma}
\begin{proof}
    Suppose not. This means that from any vertex~$v_1$, there are vertices up to distance $\Theta(k\log n/\phi_k)$. Select vertices $v_2, v_3, \dots, v_k$ such that $d(v_i,v_j)>\Theta(\log n/\phi_k)$ for $i\neq j$. Now, for any $v_i$, suppose for contradiction that there are no cuts $S$ within a distance $\Theta(\log n/\phi_k)$ with $\phi(S)\le \phi_k$. Then for any $j\le \log n/\phi_k$, the $j$-hop neighborhood $N_j(v_i)$ satisfies $|E(N_j(v_i), V\setminus N_j(v_i))|\ge \phi_k\Vol(N_j(v_i))$, so within $\log_{1+\phi_k}n=O(\log n/\phi_k)$ steps, we cover the entire graph, a contradiction. 
\end{proof}



 \subsection{\CONGEST Upper Bound for Approximating the Sparsest Cut}
 We combine the Cheeger's Inequality with our eigenvalue approximation to obtain the following result. 
\CONGESTSparsestCut*
\begin{proof}
    By \Cref{thm:CONGEST_smallest}, we can find $\tilde\lambda_2$ such that 
    \begin{equation*}
        \lambda_2 \le \tilde \lambda_2 \le \lambda_2+\delta.
    \end{equation*}
    This takes $O\left(\tfrac{\log^2 n}{\delta} + \tfrac{\log n}{\delta}\cdot \left\lceil \tfrac{\lambda_2}{ \delta}\right\rceil\right)$ rounds. 

    Now we output $\tilde \phi_2=  \sqrt{2\tilde \lambda_k}$. 
    Using Cheeger's Inequality, \Cref{thm:Cheeger}, we see that 
    \begin{align*}
        \tilde \phi &= \sqrt{2\tilde \lambda_k} \ge \sqrt{2 \lambda_k} \ge \phi
        \intertext{and that}
         \tilde \phi &= \sqrt{2\tilde \lambda_2} \le \sqrt{ 2\lambda_2+2\delta} \le  \sqrt{ 2\phi+2\delta}. 
         \intertext{When we set  $\delta = \eps^2/2$ we see that }
         \tilde \phi &\le  \sqrt{ 2\phi+2\delta} =  \sqrt{ 2\phi+\eps^2} \le \sqrt{2\phi}+\eps.
    \end{align*} 
    Hence we can attain the $\sqrt{n\phi}$-bound of Cheeger's Inequality up to an arbitrarily small additive error~$\eps$ in $O(\log^2 n /\eps^2+\log n/\eps^4)$ rounds.
    When we set $\delta = 0.005\phi$, we get the stated approximation: $\tilde \phi \le \sqrt{2.01\phi}$. This takes 
    \begin{align*}
        O\left(\tfrac{\log^2 n}{\delta} + \tfrac{\log n}{\delta}\cdot \left\lceil \tfrac{\lambda_2}{ \delta}\right\rceil\right) &= O\left(\tfrac{\log^2 n}{\phi} + \tfrac{\log n}{\phi}\cdot \left\lceil \tfrac{\lambda_2}{ \phi}\right\rceil\right)
    \end{align*}
    rounds. Cheeger's Inequality states that $\tfrac{\lambda_2}{2}\le \phi$, so $\tfrac{\lambda_2}{\phi}\le 2$. This simplifies the running time to
    $ O\left(\tfrac{\log^2 n}{\phi} + \tfrac{\log n}{\phi}\right) = O\left(\tfrac{\log^2 n}{\phi} \right) $
    rounds.     
    
    Note that, of course, $\phi$ is not known at the beginning of the algorithm, but we can approximate it without additional asymptotic costs via an exponential search with $\phi=1/2, 1/4, \dots$.  
\end{proof}

 \subsection{\CONGEST Upper Bound for Approximating the \texorpdfstring{$k$}{k}-Way Sparsest Cut}
  We combine the generalized Cheeger's Inequality with our eigenvalue approximation to obtain the following result. 
 \CONGESTkWaySparsestCut*
\begin{proof}


        
    By \Cref{thm:congest_lambda_k}, we can find $\tilde\lambda_k$ such that 
    \begin{equation*}
        \lambda_k \le \tilde \lambda_k \le \lambda_k+\eps.
    \end{equation*}
    This takes $ O(\log^2 n \poly(\eps^{-1}))$ rounds. 

    Now we output $\tilde \phi_k= \Theta(k^2) \sqrt{\tilde \lambda_k}$. 
    Using the higher order Cheeger's Inequality, \Cref{thm:HO_Cheeger}, we see that 
    \begin{align*}
        \tilde \phi_k &= \Theta(k^2) \sqrt{\tilde \lambda_k} \ge \Theta(k^2) \sqrt{ \lambda_k} \ge \phi_k
        \intertext{and that}
         \tilde \phi_k &= \Theta(k^2) \sqrt{\tilde \lambda_k} \le\Theta(k^2) \sqrt{ \lambda_k+\eps} \le \Theta(k^2) \sqrt{ 2\phi_k+\eps}. 
    \end{align*}
    Setting $\eps=\phi_k$ gives the result. 
\end{proof}

 \subsection{Lower Bounds}
As before, we assume that the graph is connected, as in standard in the \LOCAL and \CONGEST model. We allow for the general case that vertices know the value $n$. 
First, we give a diameter lower bound for all possible diameters. 

\SparsestCutMultLB*
\begin{proof}
    We consider two types of graphs with $n$ vertices and diameter $D$. If with probability $>\tfrac{1}{2}$ it is correct on each, then there has to exist output that is correct on both. We consider this case. Let $N$ denote  $(n-D+2)/2$. 
    
    Let $G_1$ be  a path of length $D-2$, connected to a clique on $N$ vertices on both sides. This graph has diameter $D$ and sparsest cut $\Theta(1/n^2)$. For the latter, we consider a cut $S$, cutting the graph anywhere on the path. The volume on both sides is $\Theta(n^2)$, since each clique has $n^2$ edges, so this gives a cut of sparsity 
    \begin{equation*}
        \phi(S) = \frac{|E(S,V\setminus S)|}{\min\{\Vol(S),\Vol(V\setminus S)\}}= \frac{1}{\min\{\Theta(n^2),\Theta(n^2)\}} = \Theta(1/n^2).
    \end{equation*}
    
     Cutting through the clique always gives a bigger value: say we have $M<N$ vertices of the clique in $S$, then 
     \begin{equation*}
        \phi(S) = \frac{|E(S,V\setminus S)|}{\min\{\Vol(S),\Vol(V\setminus S)\}}= \frac{M(N-M)}{\min\{M(N-1),\Theta(n^2)\}} = \frac{N-M}{N-1} \ge \frac{1}{N-1}=\Theta(1/n).
    \end{equation*}
   
    Second, consider $G_2$,  a path of length $D$, connected to a clique on $N$ vertices on one side, and connected to the center of a star graph on $N$ vertices on the other side. This graph also has diameter $D$, but the sparsest cut is $\Theta(1/n)$. As before, cutting through the clique gives $\phi(S)\ge \Theta(1/n)$. When we cut anywhere not on the path, we are cutting a tree. To get the biggest volume (and hence smallest sparsity), we cut the first edge after the clique 
    \begin{equation*}
        \phi(S) = \frac{|E(S,V\setminus S)|}{\min\{\Vol(S),\Vol(V\setminus S)\}}= \frac{1}{\min\{D+N,\Theta(n^2)\}} = \frac{1}{\min\{\Theta(n),\Theta(n^2)\}}= \Theta(1/n). 
    \end{equation*}

    No vertex in $G_1$ can discern between those two graphs in less than $\Theta(D)$ time: a vertex can only see at most one clique. The other side could be the star.     
    So every vertex $v$ needs to output $\tilde \phi_v \ge \Theta(1/n)$, but this is only a $\Theta(n)$ approximation for $\phi(G_1)=\Theta(1/n^2)$.  
%
 %
\end{proof}

Second, we give a diameter lower bound for most values of $\phi$. 
\SparsestCutAddLB*
\begin{proof}
    Fix $\phi > 1/n^{1-\eps}$.    
    Chierichetti et al.~\cite{ChierichettiGLP18} showed that there exist graphs with sparsest cut $\phi$ and with diameter $D=\Theta(\log n/\phi)$. Consider such a graph $G$. Say that after $r=D/3$ rounds every vertex outputs $\tilde \phi_v$ such that $\phi \le \tilde \phi_v \le \Theta(n^{1+\eps}\cdot \phi)$. 
    If every vertex has more than $n/2$ vertices in its $r$-hop neighborhood, then the diameter is at most $2r<D$, a contradiction.     
    Consider a vertex $v\in V$ with most $n/2$ vertices in its $r$-hop neighborhood. By assumption $v$ needs to output $\tilde \phi_v \ge \phi$. 

    Now consider a graph $G'$, consisting of $v$ and its $r$-hop neighborhood as before, but the remaining vertices form 2 cliques on $n/4$ vertices, $C_1$ and $C_2$. $C_1$  connected to $C_2$ by one edge, and $C_2$ is connected to the $r$-hop neighborhood of $v$. This graph has sparest cut at most $\Theta(1/n^2)$. Indeed, consider the cut $S$ between $C_1$ and $C_2$: this has one edge, and volume $\Theta(n^2)$ on each side, so $\phi(S)\le \Theta(1/n^2)$. By assumption, $v$ needs to output $\tilde \phi_v \le  \Theta(n^{1+\eps} \cdot 1/n^2)= \Theta(1/n^{1-\eps})<\phi$, for appropriately chosen constants.

    If with probability $>\tfrac{1}{2}$ the algorithm is correct on each $G$ and $G'$, then there has to exist output that is correct on both. We consider this case. 

    Since $v$ cannot distinguish between these two graphs, it needs to output $\tilde \phi_v \ge \phi$ \emph{and} $\tilde \phi_v < \phi$, a contradiction. 
\end{proof}

\newpage
  \printbibliography[heading=bibintoc]

\appendix
  \newpage

\section{The Power Method}
\label{app:powerMethod}
The power method is a standard technique in numerical linear algebra, see, e.g., \cite{golub2013matrix}. Although it works for any type of matrix, both its statement and analyses are simpler for positive semi-definite instances. Since this suffices fo our purposes, we restrict to this case. In \Cref{ssec:PowerMethodProof}, we prove the known \Cref{lm:powerMethod}, following~\cite{trevisan-expanders}.
In \Cref{ssec:bitPrecision} we prove the bit precision lemma (\Cref{lm:num_stab}) required for an efficient CONGEST implementation. 
\subsection{The Power Method (Proof of Lemma~\ref{lm:powerMethod})}
\label{ssec:PowerMethodProof}

\paragraph{Algorithm.} 
The power method starts with initial vector $x_0\sim \{-1,1\}^n$, chosen uniformly at random vector. 
We then compute $x_k := M^k x_0$. The goal is to show that the Rayleigh coefficient of $x_k$ approximates the largest eigenvalue $\mu_n$ well: 
    \begin{equation*}
        (1-\eps)\mu_n\le\frac{x_k^TMx_k }{x_k^T x_k} \le \mu_n.
    \end{equation*}

\paragraph{Start vector.}
Notably, the start vector $x$ has to have a significant inner product with the largest eigenvector~$v_n$.
Choosing $x$ uniformly at random ensures that this is true with constant probability. For this proof, we use the Paley-Zygmund Inequality.

\begin{lemma}[Paley-Zygmund Inequality]
    \label{lm:paley_zygmund_inequality}
    Let $X \ge 0$ be a random variable with finite variance.
    For all $0 \le \delta \le 1$ holds that
    \begin{equation*}
        \P[X \ge \delta \cdot \E(X)] \ge (1 - \delta)^2 \cdot \frac{\E(X)^2}{\E(X^2)}.
    \end{equation*}
\end{lemma}

The power method fails to work if the random start vector is orthogonal to the eigenvector of the largest eigenvalue. The following lemma states that it is unlikely to happen. 
\begin{restatable}{lemma}{lemStartVector}\label{lm:start_vector}
    Let $x\sim \set{-1,1}^n$ be a uniformly random vector and $y \in \R^n$ a vector with $\norm{y} = 1$.
    Then we have
    \begin{equation*}
        \P\left[\abs{\innerprod{x}{y}}\ge \tfrac{1}{2}\right] \ge \tfrac{3}{16}.
    \end{equation*}
\end{restatable}
\begin{proof}
    Let $X = \sum_{i \in [n]} x_i y_i$ be the random variable defined by the inner product $\innerprod{x}{y}$.
    As $X$ is not necessarily non-negative, we base our calculation on $X^2$ which is always non-negative.
    First, we consider the second moment of $X$.
    \begin{align*}
        \E(X^2) &= \E \left(\left(\sum_{i \in [n]} x_i y_i\right)^2 \right)
        = \sum_{i,j \in [n]} y_i y_j \E(x_ix_j)
        =\sum_{i \in [n]} y_i^2\E(x_i^2) + \sum_{i, j\in [n], i \neq j} y_iy_j \E(x_ix_j) \\
        &=\sum_{i \in [n]} y_i^2  + \sum_{i, j\in [n], i \neq j} y_iy_j
        = \sum_{i \in [n]} y_i^2 = 1
    \intertext{In the fourth step, we apply $\E(x_ix_j) = 0$ and $\E(x_i^2) = 1$.
    The last equality follows from the fact that $\norm{y} = 1$.
    Next, consider the fourth moment of $X$.}
        \E(X^4) 
        &= 3 \cdot \sum_{i,j \in [n], i \neq j} \E(x_i^2)y_i^2 \cdot \E(x_j^2)y_j^2 + \sum_{i \in [n]} \E(x_i^4)y_i^4 
        =  3 \cdot \sum_{i,j \in [n]} y_i^2 y_j^2 - 3 \cdot \sum_{i \in [n]} y_i^4 + \sum_{i \in [n]} y_i^4  \\
        &= 3\cdot \left(\sum_{i \in [n]} y_i^2 \right)^2 - 2 \cdot \sum_{i \in [n]} y_i^4 
        = 3 - 2 \cdot \sum_{i \in [n]} y_i^4 
        \le 3. 
    \end{align*}
    For the sake of simplicity, the first step leaves out any summations that contain $\E(x_i)$ as all of them are equal to zero.
    In the second step, we use the fact that $\E(x_i^2) = 1 = \E(x_i^4)$.
    
    Now, we are in the position to apply the Paley-Zygmund inequality.
    Let $\delta = 1/4$.
    As $X^2$ is a non-negative random variable with finite variance, it follows from  \cref{lm:paley_zygmund_inequality} that
    \begin{equation*}
        \P[X^2 \ge \delta \cdot \E(X^2)] 
        = \P[X^2 \ge \tfrac{1}{4}]
        \ge (1 - \tfrac{1}{4})^2 \cdot \frac{\E(X^2)^2}{\E(X^4)}
        = (\tfrac{3}{4})^2 \cdot \tfrac{1}{3}
        = \tfrac{3}{16}.\qedhere
    \end{equation*}
\end{proof}

\paragraph{Lower Bound.}
\Cref{lm:powerMethod} states that the Rayleigh coefficient output by the power method is lower bounded by $ (1-\eps)\mu_n$. Next, we prove a slightly weaker lower bound; the original lower bound is obtained in the proof of \Cref{lm:powerMethod} by adjusting $\eps$.
\begin{lemma}\label{lm:largest_eigenvalue_ratio}
    Let $x_0 \in \R^n$ and set $x_k = M^kx$.
    For any $\eps > 0$ and $k \ge \log n / \eps$, we have
    \begin{equation}
        \label{eq:largest_eigenvalue_ratio}
        \frac{x_k^TMx_k }{x_k^T x_k} \ge (1-\eps)\mu_n \cdot \frac{1}{1 + \frac{\norm{x}^2}{\innerprod{x_0}{v_n}^2}(1- \eps)^{2k}}.
    \end{equation}
\end{lemma}
\begin{proof}
    We start by dividing the eigenvalues $ \mu_1\le \mu_2 \le \dotsc \le \mu_n$ of~$M$ in two groups.
    Let $j \in [n]$ be the first index for which $(1-\eps)\mu_n \le \mu_{j}$ holds, then we have $\mu_1, \dotsc ,\mu_{j-1} < (1-\eps)\mu_n$ and $(1-\eps)\mu_n \le \mu_{j}, \dotsc ,\mu_n$.
    As a direct consequence, it holds that 
    \begin{equation*}
        \sum_{i \in [j-1]} \mu_i^{2k} \le n (1 - \eps)^{2k}\mu_n^{2k} \le \frac{\mu_n^{2k}}{n}.
    \end{equation*}
    Intuitively, the contribution of the small eigenvalue group to the spectral decomposition of $M^{2k+1}$ is insignificant for large $k$.
    With this insight, we use the spectral decomposition of~$M^{2k+1}$ to derive a lower bound for the nominator
    \begin{align*}
        x_k^TMx_k 
        = x_0^T M^{2k+1} x_0 
        &= x_0^T \left( \sum_{i \in [n]} \mu_i^{2k+1}v_iv_i^T \right) x_0 \\
        &= \sum_{i \in [n]} \mu_i^{2k+1} \innerprod{x_0}{v_i}^2 \\
        &\ge \sum_{i \in \set{j, \dotsc, n}} \mu_i^{2k+1} \innerprod{x_0}{v_i}^2 \\
        &\ge (1-\eps)\mu_n  \sum_{i \in \set{j, \dotsc, n}} \mu_i^{2k} \innerprod{x_0}{v_i}^2.
    \intertext{Second, we derive an upper bound for the denominator}
        x_k^Tx_k 
        = x_0^T M^{2k} x_0
        &= \sum_{i \in [n]} \mu_i^{2k} \innerprod{x_0}{v_i}^2 \\
        &= \sum_{i \in \set{j, \dotsc, n}} \mu_i^{2k} \innerprod{x_0}{v_i}^2 + \sum_{i \in \set{1, \dotsc, j-1}} \mu_i^{2k} \innerprod{x_0}{v_i}^2  \\
        &\le \sum_{i \in \set{j, \dotsc, n}} \mu_i^{2k} \innerprod{x_0}{v_i}^2 + (1-\eps)^{2k}\mu_n^{2k} \sum_{i \in \set{1, \dotsc, j-1}} \innerprod{x_0}{v_i}^2 \\
        &\le \sum_{i \in \set{j, \dotsc, n}} \mu_i^{2k} \innerprod{x_0}{v_i}^2 + (1-\eps)^{2k}\mu_n^{2k} \norm{x_0}^2,
    \intertext{where the last inequality follows from $\sum_{i \in [j-1]} \innerprod{x_0}{v_i}^2 \le \norm{x_0}^2$. 
    Using both upper and lower bound, we can bound the ratio as follows}
        \frac{x_k^TMx_k }{x_k^T x_k}
        = \frac{x_0^T M^{2k+1} x_0}{x_0^T M^{2k} x_0}
        &\ge \frac{(1-\eps)\mu_n \sum_{i \in \set{j, \dotsc, n}}\mu_i^{2k} \innerprod{x_0}{v_i}^2}{(1-\eps)^{2k}\mu_n^{2k} \norm{x_0}^2 + \sum_{i \in \set{j, \dotsc, n}} \mu_i^{2k} \innerprod{x_0}{v_i}^2} \\
        &= \frac{(1-\eps)\mu_n}{1 + \frac{(1-\eps)^{2k}\mu_n^{2k} \norm{x_0}^2}{\sum_{i \in \set{j, \dotsc, n}} \mu_i^{2k} \innerprod{x_0}{v_i}^2}} \\
        &\ge \frac{(1-\eps)\mu_n}{1 + \frac{(1-\eps)^{2k}\mu_n^{2k} \norm{x_0}^2}{\mu_n^{2k} \innerprod{x_0}{v_n}^2}} \\
        &= \frac{(1-\eps)\mu_n}{1 + (1-\eps)^{2k}\frac{\norm{x_0}^2}{\innerprod{x_0}{v_n}^2}}.
        \qedhere
    \end{align*}
\end{proof}

\paragraph{Upper Bound.}
Next, we show that we cannot overshoot, i.e, that the output of the power method is upper bounded by $\mu_n$. 
\lemNoOvershoot*
\begin{proof}
    We decompose $x$ into the eigenvectors, i.e., let $\alpha_j\in \R$ be $\alpha_j:=\langle x,v_j\rangle$, and hence $x= \sum_{j\in [n]}\alpha_j v_j$. Now we observe 
    \begin{align*}
        x^TMx &= \left( \sum_{j\in [n]}\alpha_j v_j\right)^T\left( M\left( \sum_{j\in [n]}\alpha_j v_j\right)\right) =\sum_{j\in [n]}\alpha_j^2 \mu_j \\
        &\le \sum_{j\in [n]}\left(\max_{k\in [n]}\mu_k\right)  \alpha_j^2  =\mu_n \sum_{j\in [n]}  \alpha_j^2 = \mu_n \left( \sum_{j\in [n]}\alpha_j v_j\right)^T\left( \sum_{j\in [n]}\alpha_j v_j\right)=\mu_n x^Tx.\qedhere
    \end{align*}
\end{proof}

\paragraph{Conclusion.}
We conclude by putting the pieces together.
\PowerMethod*
\begin{proof}
    Due to \cref{lm:start_vector} and $\norm{x}^2 = n$, we have $\norm{x}^2/\innerprod{x}{v_1}^2 \le 4n$.
    Thus, the right denominator in \cref{eq:largest_eigenvalue_ratio} can be bounded by
    \begin{equation*}
        1 + \frac{\norm{x}^2}{\innerprod{x}{v_1}^2}(1- \eps)^{2k} 
        \le 1 + 4n \cdot e^{- 2 \log n}
        \le 1 + \tfrac{4}{n}.
    \end{equation*}
    For any $\eps \ge 8/n$ and $\eps' = 2\eps - \eps^2$, \cref{eq:largest_eigenvalue_ratio} simplifies to
    \begin{equation*}
        \frac{x_k^TMx_k }{x_k^T x_k} 
        \ge (1-\eps)\mu_n \cdot \frac{1}{1 + \tfrac{4}{n}} 
        \ge (1-\eps)\mu_n \cdot (1 - \tfrac{8}{n})
        \ge (1-\eps')\mu_n.
    \end{equation*}
    By \Cref{lm:no_overshoot}, we get the upper bound. 
\end{proof}


    

\subsection{Bit Precision Lemma (Proof of Lemma~\ref{lm:num_stab})}
\label{ssec:bitPrecision}

The following lemma shows that truncating or rounding after the $\Theta(\log n)$-th bit only impact the outcome by an inverse polynomial factor. This means we can round up to $\poly(n)$ times without any issues. In particular, we can round in every step of the algorithm with a negligibly small effect on the outcome. 

\lemNumStab*
\begin{proof}
    The power method starts with $x_0\sim \{+1,-1\}$ uniformly at random, so  $||x_0||=\sqrt n$, and with constant probability, \Cref{lm:start_vector}, we have $|\langle x_0,v_n\rangle| \geq \tfrac{3}{16}$. In which case the power method gives a good approximation. We investigate the numerical stability of this case.

    We need to compare before and after rounding for 
    \begin{equation*}
        \frac{(x_k)^TMx_k }{(x_k)^Tx_k}= \frac{x_0^TM^{2k+1}x_0}{x_0M^{2k}x_0}= \frac{x_i^TM^{2l+1}x_i}{x_iM^{2l}x_i},
    \end{equation*}  
    where $l:=k-i$. Denote the difference occurred by rounding as $y$, i.e., after rounding we have $x_i+y$ and since we round after $\Theta(\log n)$ bits, we have $|(y)_v| \le \delta |(x_i)_v|$ for some $\delta = n^{-\Theta(1)}$. In particular $||y|| \le \delta ||x_i||$. 

    Now, let use make some observations. First of all, 
    \begin{align*}
        |\langle x_i,v_n\rangle| &\ge \mu_n^i |\langle x_0,v_n\rangle| \ge \mu_n^i \left(\frac{3}{16}\right)^2.\\
        ||x_i|| &\le \mu_n^i ||x_0||=\mu_n^i \sqrt n. 
    \end{align*}

    Next, we use this to bound our main inner product:
    \begin{align*}
        x_i^TM^{2l+1}x_i &\ge \mu_n^{2l+1}\langle x_i,v_n\rangle^2 \ge \mu_n^{2k+1}\cdot \frac{3}{16}.
        \intertext{Or equivalently, $\mu_n^{2k+1} \le \Theta(1) x_i^TM^{2l+1}x_i$. Next, we use Cauchy-Schwarz to express the incurred error in terms of this inner product.}
        |y^TM^{2l+1}x_i| &\le ||y||\cdot ||M^{2l+1}x_i|| \le \delta ||x_i||\mu_n^{2l+1} \le \delta \mu^{2k+1}\sqrt n\le \Theta(\sqrt n) x_i^TM^{2l+1}x_i\\
        |x_i^TM^{2l+1}y| &\le  ||M^{2l+1}x_i||\cdot ||y|| \le \delta ||x_i||\mu_n^{2l+1} \le \delta \mu^{2k+1}\sqrt n\le \Theta(\sqrt n) x_i^TM^{2l+1}x_i\\
        |y^TM^{2l+1}y| &\le  ||y||^2\mu_n^{2l+1} \le \delta^2||x_i||\mu_n^{2l+1}\le\delta^2 \mu_n^{2k+1}\sqrt n\le\delta\mu_n^{2k+1}\le \Theta(\sqrt n) x_i^TM^{2l+1}x_i,
    \end{align*}

Now, we are ready to consider $\frac{(x_i+y)^TM^{2l+1}(x_i+y)}{(x_i+y)M^{2l}(x_i+y)}$. 
   We do this in two steps: numerator and denominator. First we note for the numerator that

    \begin{align*}
      (x_i+y)^TM^{2l+1}(x_i+y)  &\ge x_i^TM^{2l+1}x_i
      - |y^TM^{2l+1}x_i|-|x_i^TM^{2l+1}y|-|y^TM^{2l+1}y| \\
      &\ge (1-3\delta\Theta(\sqrt n))    x_i^TM^{2l+1}x_i.
    \end{align*}

    For the denominator, we obtain via similar arguments that 
    \begin{equation*}
       (x_i+y)M^{2l}(x_i+y) \le (1+3\delta\Theta(\sqrt n)) x_iM^{2l}x_i. 
    \end{equation*}

    So in total we have 

    \begin{equation*}
        \frac{(x_i+y)^TM^{2l+1}(x_i+y)}{(x_i+y)M^{2l}(x_i+y)}  \ge \frac{1-3\delta\Theta(\sqrt n)}{1+3\delta\Theta(\sqrt n)}\cdot \frac{x_i^TM^{2l+1}x_i}{x_iM^{2l}x_i}. 
    \end{equation*}
    For $\delta \le n^{-\Theta(1)}$, we obtain 
    \begin{equation*}
        \frac{(x_i+y)^TM^{2l+1}(x_i+y)}{(x_i+y)M^{2l}(x_i+y)}  \ge (1-n^{-\Theta(1)})\cdot \frac{x_i^TM^{2l+1}x_i}{x_iM^{2l}x_i}. 
    \end{equation*}

    Similar arguments give the upper bound with $(1+n^{-\Theta(1)})$. 
\end{proof}





\end{document}